\documentclass[journal]{IEEEtran}
\usepackage{graphicx}
\usepackage[cmex10]{amsmath}
\usepackage{cases}
\usepackage{amsthm}
\usepackage{cite}
\usepackage{amssymb}
\usepackage{algorithm}
\usepackage{algorithmic}
\usepackage{multirow}
\usepackage{amsmath}
\usepackage{xcolor}
\usepackage{subeqnarray}
\usepackage{cases}
\usepackage{enumerate}
\usepackage{cuted}
\usepackage{booktabs}
\usepackage{mathrsfs}
\usepackage{bbm}
\usepackage{graphicx}
\usepackage{array}
\newcolumntype{M}[1]{>{\centering\arraybackslash}m{#1}}
\usepackage[colorlinks]{hyperref}

\newtheorem{theorem}{\hskip\parindent \it Theorem}
\newtheorem{remark}{\hskip\parindent\bf Remark}

\ifCLASSINFOpdf
\else
\fi

\begin{document}
\title{Multi-Subarray FD-RIS Enhanced Multi-user Wireless Networks: With Joint Distance-Angle Beamforming}
\author{Han Xiao,~\IEEEmembership{Student Member,~IEEE,} Xiaoyan Hu,~\IEEEmembership{Member,~IEEE,} Wenjie Wang,~\IEEEmembership{Senior Member,~IEEE,} \\
Kai-Kit~Wong,~\IEEEmembership{Fellow,~IEEE}, Kun~Yang,~\IEEEmembership{Fellow,~IEEE}, Shi Jin,~\IEEEmembership{Fellow,~IEEE}
	
\thanks{H. Xiao, X. Hu, and W. Wang are with the School of Information and Communications Engineering, Xi'an Jiaotong University, Xi'an 710049, China. (email: hanxiaonuli@stu.xjtu.edu.cn, xiaoyanhu@xjtu.edu.cn, wjwang@mail.xjtu.edu.cn). \emph{(Corresponding author: Xiaoyan Hu.)}} %
\thanks{K.-K. Wong is with the Department of Electronic and Electrical Engineering, University College London, London WC1E 7JE, U.K. (email: kai-kit.wong@ucl.ac.uk).}
\thanks{K. Yang is with the School of Computer Science and Electronic Engineering, University of Essex, Colchester CO4 3SQ, U.K. (e-mail: kunyang@ essex.ac.uk).}
\thanks{S. Jin is with the National Mobile Communications Research Laboratory, Southeast University, Nanjing 210096, China. (e-mail: jinshi@seu.edu.cn,).}
}

\maketitle
\begin{abstract}
	The concept of the frequency diverse reconfigurable intelligent surface (FD-RIS) technology has been introduced, which can enable simultaneous implementation of distance-angle beamforming in far-field communication scenarios. In order to improve the managing ability on undesired harmonic signals and the diversity of frequency offsets, this paper presents a novel multi-subarray FD-RIS framework. In this framework, the RIS is evenly divided into multiple subarrays, each employing a distinct time-modulation frequency to enable the diversity of frequency offsets. Additionally, to suppress the undesired harmonic signals, a new time-modulation technique is employed to periodically adjust the phase-shift of each element. In this case, the signal processing model is first analytically derived. Then, we integrate it into a multi-user communication scenario and formulate an optimization problem that aims to maximize the weighted sum rate of all users. This is achieved by jointly optimizing the active beamforming, time delays, and modulation frequencies. Subsequently, a novel iterative algorithm is proposed to effectively solve this problem with low computing complexity. \textcolor{blue}{To evaluate the beamforming capability of the proposed multi-subarray FD-RIS, three communication scenarios with varying	spatial correlations among users are considered. Simulation results demonstrate that the proposed multi-subarray FD-RIS can significantly enhance and sustain the performance of communication networks by leveraging unique distance-angle beamforming, even when users share the same angular position where traditional RIS experiences severe degradation.} The proposed algorithm also demonstrates a notably superiority in performance and computational complexity compared  with the baseline algorithms such as semi-definite relaxation (SDR) and zero-forcing (ZF).
\end{abstract}
\IEEEpeerreviewmaketitle

\section{Introduction}\label{sec:S1}
The international telecommunication union (ITU) has comprehensively defined the application scenarios and performance metrics for future 6G networks in IMT-2030 \cite{recommendation2023framework}. Among these, a significant emphasis has been placed on the immense demand for spectral efficiency and system capacity in communication systems. However, the spectral efficiency and system capacity of existing communication systems are gradually approaching their theoretical maximums. As a result, a significant challenge for upcoming 6G communication networks is to find new ways to improve these metrics further. \textcolor{blue}{As a key technology for 5G and future 6G networks, massive multiple-input multiple-output (MIMO) can significantly enhance the spectral efficiency of communication systems \cite{andrews2014will} by deploying a large number of antennas at the base station (BS). Leveraging these antennas, massive MIMO achieves strong beamforming capability, which in turn greatly improves energy efficiency of systems and interference management ability among users.}
 However, implementing massive MIMO technology in real-world applications faces several challenges, including high costs, high hardware complexity and significant energy consumption \cite{zhang2016fundamental}. These challenges stem primarily from the need for expensive radio frequency (RF) components and advanced signal processing systems to achieve optimal performance, which poses a substantial barrier to the practical deployment of the massive MIMO in future 6G networks.

\subsection{Related Work}
Reconfigurable intelligent surface (RIS) has emerged as a pivotal technology for future 6G networks, providing a cost-effective, hardware-efficient, and energy-efficient solution to mitigate the challenges faced by massive MIMO technology. By utilizing low-cost passive elements and simple control circuitry, RIS can effectively modify the electromagnetic properties of incoming signals, thereby enabling powerful beamforming capability \cite{wu2019intelligent, huang2019reconfigurable,nemati2020ris, zhou2022channel}. Specifically, \cite{wu2019intelligent} and \cite{huang2019reconfigurable} provide the systematic exploration for the RIS's potentials in enhancing energy efficiency of multi-user communication systems, conducting the research from different optimization objectives. The focus of \cite{wu2019intelligent} is on minimizing  total energy consumption of the system while meeting all users' quality of service (QoS) constraints, whereas \cite{huang2019reconfigurable} aims to optimize energy efficiency within a given power budget. Simulation results demonstrate that compared to traditional active relay systems, RIS significantly improves the energy efficiency of communication systems
and shows superior balancing capabilities in terms of users' QoS. Furthermore, the simulation results in \cite{wu2019intelligent} indicate that by achieving similar performance gains as massive MIMO systems, integrating RIS can significantly reduce the required number of antennas at BSs, thus lowering the hardware complexity and costs of the system. Subsequently, the authors in \cite{nemati2020ris} explore the feasibility of the RIS to improve the coverage of mmWave networks. In particular, they derive the closed-form expression for the peak reflection power of the RIS and the downlink signal-to-interference ratio (SIR) coverage. The results demonstrate the significant advantages of RIS in enhancing coverage for mmWave communication networks. Based on the aforementioned advantages, RIS has been incorporated into various communication networks to enhance performance  such as secure communications \cite{Cui2019Secure, Xiao2024simultaneously, xiao2024STAR-RIS}, mobile edge computing (MEC) \cite{hu2021reconfigurable, he2024Joint, xiao2025STAR-RIS_UAV, xiao2025Energy-Efficient}, and integrated sensing and communications (ISAC) \cite{zhu2023joint, zhong2023joint}.

\textcolor{blue}{Actually, the strong signal control capability of RIS is significantly constrained in certain communication scenarios \footnote{Unless stated otherwise, the system studied in this paper operates under the far-field communication model.}, particularly when some users exhibit high spatial correlation or even overlap in spatial directions. This limitation stems from the fact that conventional RIS can perform beamforming for incident signals only in the angular domain due to the plane-wave characteristics of electromagnetic propagation, which poses significant challenges in enabling independent beamforming or differentiated signal control for such users. Therefore, the single-dimensional control of conventional RIS lacks the adaptability and robustness for multi-user scenarios. To address this issue, introducing new beamforming degrees of freedom (DoFs) represents a promising research direction.}
 Note that \cite{shen2023multi} highlights that the RIS is able to flexibly manipulate and concentrate signal energy from both distance and angle dimensions in \textcolor{blue}{near-field} communications due to the unique spherical wave properties of electromagnetic waves. This dual-dimensional control significantly improves spatial resolution and energy efficiency. Motivated by the beamforming capabilities, integrating distance-based beamforming into RIS holds significant potential for overcoming the aforementioned challenges. However, realizing distance-based beamforming inherently difficult due to the planar propagation characteristics of electromagnetic waves. Therefore, a critical question arises: ``\textbf{How can distance-based beamforming be effectively incorporated into RIS's beamforming?}'' It is worth noting that addressing this challenge is essential for unlocking new possibilities in RIS-supported communication systems.

Fortunately, the frequency diverse array (FDA) antenna serves as a prime example of achieving effective distance-angle beamforming in antenna systems \cite{antonik2006frequency}. Specifically, before transmitting the signal, a slight frequency offset is introduced to the carrier of each antenna, through which the distance-dimension beamforming can be effectively incorporated into the signal's transmitting. Actually, FDA antennas were originally applied in radar systems to enhance their detection capabilities \cite{khan2014frequency, xiong2016frequency}. Recently, some researchers begin to introduce FDA technology into communication systems \cite{wang2018retrodirective, liu2023rate, nusenu2022power} to explore its potentials in enhancing communication performance. \cite{wang2018retrodirective} initially integrates the FDA technology into simultaneous wireless information and power transfer systems, and extensive numerical results demonstrate that FDA can significantly outperform phased arrays in enhancing users' achievable rates and energy harvesting capabilities. Inspired by the intrinsic mechanism of the FDA antennas, if RIS can integrate different frequency offsets into the incident signals, it has the potential to achieve joint distance-angle beamforming for the incoming signals. However, unlike FDA antennas, RIS, as a passive device, lacks the ability to attain signal frequency diversity through RF components.

To address this issue, the frequency diverse RIS (FD-RIS) proposed in \cite{xiaoFD-RIS2024} provides a feasible solution. Its core idea can be summarized as follows: when time modulation techniques are used to periodically alter the phase shifts of  RIS,  the incident signal will be scattered into a series of harmonic signals according to Fourier series expansion theory. Each harmonic frequency has a frequency offset relative to the centre carrier frequency, and the size of this offset is determined by the harmonic order and the time modulation period. Based on this principle, the authors of \cite{xiaoFD-RIS2024} have verified the joint distance-angle beamforming capability of the FD-RIS through both theoretical analysis and numerical simulations. 

\subsection{Motivation and Contributions}
Although the effectiveness of the FD-RIS in combining distance and angle beamforming for signal transmission has been verified, the investigation of the application potentials for FD-RIS-assisted communication systems is still in its infancy stage. In fact, research on FD-RIS in \cite{xiaoFD-RIS2024} has made some progress, yet there are still certain challenges that need to be addressed, primarily in the following aspects: (\romannumeral 1) The potential of the FD-RIS in single-user scenarios has been examined, however, it is crucial to consider the more prevalent multi-user scenarios in practical applications. Multi-user communication systems face even more complex challenges, especially in environment where the spectrum sharing and interference management issues significantly impact the system performance. Therefore, investigating the capability of the FD-RIS in enhancing the spectrum efficiency of multi-user communication systems holds significant theoretical importance. (\romannumeral 2) In the exiting FD-RIS system, all elements modulate signals using identical period, enabling the FD-RIS only provide linear frequency shifts for incoming signals. Clearly, the lack of diversity in frequency shifts may constrain the freedom of the FD-RIS in distance and angle control, thereby impacting its potential applications in complex scenarios.
(\romannumeral 3) In \cite{xiaoFD-RIS2024}, the proposed FD-RIS system utilizes only a limited number of low-order harmonics for signal transmission, while neglecting the regulation of other harmonics, which may result in the unintended propagation of other harmonic signals and cause disruption to other systems.

Hence, the primary motivations of this paper are to tackle the aforementioned challenges and enhance the development of the FD-RIS. In response to these above shortcomings, a multi-subarray FD-RIS framework is proposed and utilized to enhance the performance of multi-user wireless networks. The contributions of this paper are summarized as below:
\begin{itemize}
	\item \textbf{\textit{Multi-subarray FD-RIS Framework Considering the Diversity of Frequency Offsets and Harmonic Suppression:}} A novel multi-subarray FD-RIS framework is proposed in this paper where the RIS is evenly segmented into multiple subarrays, each subarray employing a distinct time-modulation frequency to enable the diversity of frequency offsets. In order to control the undesired harmonic signals, a novel time-modulation method is leveraged to periodically adjust the phase-shifts of FD-RIS. By doing so, the undesired harmonic signals can be effectively suppressed.  However, this modulation method severely restricts the spatial control ability of the FD-RIS for the desired harmonic signals. To overcome this obstacle, different time delays are incorporated into the time-modulation procedure of each component in order to amplify the spatial degrees of freedom (DoFs) of the FD-RIS for the intended harmonic signals.
     \item\textbf{\textit{Signal Processing Model of Multi-subarray FD-RIS and Problem Formulation:}} We establish the signal processing model of the proposed multi-subarray FD-RIS and demonstrate its distance-angle beamforming capability from a theoretical perspective. To validate the effectiveness of the proposed multi-subarray FD-RIS, we incorporate it into the multi-user wireless communication scenario and formulate an optimization problem with the aim of maximizing all users' weighted sum rate through jointly designing the active beamforming of the BS, the modulation frequencies and time delays of the FD-RIS.
	 \item \textbf{\textit{Effective Alternating Algorithm with Guaranteed Convergence and Low Computational Complexity:}} It is  worthy note that the formulated optimization problem is non-convex due to the non-convex objective function and the strong couplings among variables. In order to address this problem, we first utilize the minimum mean-square error (MMSE) technique to transform the original problem, then the converted problem are divided into active beamforming subproblem, time-delay subproblem and modulation frequency subproblem. The active beamforming subproblem is tackled using a combination of the Lagrange multiplier method and bisection search method to achieve semi-closed form solutions. The time-delay subproblem is resolved through utilizing the Riemannian conjugate gradient algorithm. As for the modulation frequency subproblem, the globally convergent method of moving asymptotes (GCMMA) algorithm is employed to effectively address it with closed-form solutions. In addition, the assured convergence of the proposed iterative algorithm can be confirmed by examining the convergence curves outlined in Section \ref{sec:S5}.
	 \item \textbf{\textit{Substantial Performance Improvement:}} \textcolor{blue}{To quantify the beamforming capability of the proposed multi-subarray FD-RIS, three communication scenarios characterized by different spatial correlation levels among users are considered. Extensive simulations are conducted and the obtained simulation results indicate that compared to traditional RIS, the proposed multi-subarray FD-RIS can substantially enhance and effectively sustain system performance through its unique distance–angle beamforming capability, particularly when users share the same spatial angle, underscoring its substantial potential for future wireless networks.} The comparison with semi-definite relaxation (SDR) and zero-forcing (ZF) algorithms validate  the superiority of the proposed algorithm in terms of computational complexity and performance enhancement.
\end{itemize}

\textcolor{blue}{ The remainder of this paper is organized as follows. Section II introduces the signal processing model of the multi-subarray FD-RIS and demonstrates its distance–angle beamforming capability. Section III presents the communication system model and formulates the optimization problem. In Section IV, an iterative algorithm is developed to jointly optimize the active beamforming, time delays, and time-modulation frequencies. Section V provides simulation results to verify both the potential of the proposed FD-RIS and the effectiveness of the developed algorithm. Finally, Section VI concludes the paper and offers some application prospects of FD-RIS.
}

\textit{Notation:} $\lceil \cdot \rceil$ denotes the ceiling operator, which rounds a number up to the nearest integer. Operator $\circ$ signifies the Hadamard product. The operations $(\cdot)^T$, $(\cdot)^*$  and $(\cdot)^H$ denote the transpose, conjugate and conjugate transpose, correspondingly.  $\operatorname{Diag}(\mathbf{a})$ denotes a diagonal matrix whose diagonal elements are composed of the vector $\mathbf{a}$ .
$\operatorname{diag}(\mathbf{A})$ refers to a vector whose components consist of the diagonal elements of matrix $\mathbf{A}$. Additionally, the symbols $|\cdot|$  and $\|\cdot\|$ are indicative of the complex modulus and complex vector modulus, respectively. $\mathcal{I}_i(\mathbf{a})$ represents the $i$-th entry of the vector $\mathbf{a}$. $\mathbf{A}^{\dagger}$ represents the pseudo-inverse of the matrix $\mathbf{A}$. $x\sim \mathcal{CN}(c, d)$ represents a circularly symmetric complex Gaussian random variable denoted by $x$, characterized by a mean of $c$ and a variance of $d$.  $(x)^{+}$ and $(x)^{-}$ denote the  operations of $\max(0, x)$ and $\max(0, -x)$, respectively.

\section{Signal Processing Model}\label{sec:S2}
In this section, we will establish the signal processing model of \textcolor{blue}{the multi-subarray FD-RIS with the uniform planar array (UPA) structure.} Specifically, it is assumed that the multi-array FD-RIS is constituted by
$L=R\times S$ subarrays, where each subarray possesses $M\times N$ elements and a time modulation period $T_l$, for $l\in\mathcal{L}\triangleq\{1,\cdots, L\}$. Hence, the total number of elements equipped at the FD-RIS is $I=I_z\times I_y$, where $I_z=R\times M$ and $I_y=S\times N$. \textcolor{blue}{When the BS with $N_\mathrm{t}$ antennas transmits the signal $
		\mathbf{x}_\mathrm{BS}(t) = [w_1 e^{j2\pi f_\mathrm{c}t}s, \cdots, w_{n_\mathrm{t}} e^{j2\pi f_\mathrm{c}t}s, \cdots, w_{N_\mathrm{t}} e^{j2\pi f_\mathrm{c}t}s]^T
$ to the FD-RIS, with $w_{\mathrm{n_\mathrm {t}}}$ representing the beamforming coefficient of the $n_\mathrm{t}$-th antenna, $f_\mathrm{c}$ denoting the centre carrier frequency, and $s$ being the narrow-band complex envelope,
 the signals transmitted by different antennas are linearly superimposed in space due to the linearity of the wireless channel. Hence, for each FD-RIS element, the received signal can be viewed as  
$
	\sum_{n_\mathrm{t}=1}^{N_\mathrm{t}} w_{n_\mathrm{t}} e^{j2\pi f_\mathrm{c}t}s.
$
Therefore, the signals reflected by the FD-RIS with multiple subarrays will be superimposed as}
\begin{align}
	x(t)=&\zeta(d_\mathrm{br}) s\sum_{n_\mathrm{t}=1}^{N_\mathrm{t}}w_{\mathrm{n_\mathrm
	t}}\sum_{i_z=1}^{I_z}\sum_{i_y=1}^{I_y}\notag\\&
e^{j2\pi f_\mathrm{c}\left(t-\frac{d^\mathrm{br}_{n_\mathrm{t}, i_z , i_y}}{c}\right)}\Upsilon_{i_z, i_y}(t),
\end{align}
where 
$\zeta(d_\mathrm{br})$ represents the large-scale path loss of the transmitted signal from the BS to the FD-RIS with $d_\mathrm{br}$ being the spatial space between the reference antenna at the BS and the reference element at the FD-RIS. 
$\Upsilon_{i_z, i_y}(t)$ indicates the time-modulation reflected coefficient for the $(i_z, i_y)$-th element of the FD-RIS. $d^\mathrm{br}_{n_\mathrm{t}, i_z , i_y}$ is the distance between the $n_\mathrm{t}$-th antenna and the $(i_z, i_y)$-th element, which is given
as $d^\mathrm{br}_{n_\mathrm{t}, i_z , i_y}=d_\mathrm{br}+\Gamma_{n_\mathrm{t},i_z, i_y}^\mathrm{br}$ \cite{lu2021aerial}, where
$\Gamma^\mathrm{br}_{n_\mathrm{t}, i_z,i_y}=\Gamma^\mathrm{br}_{i_z,i_y}+(n_\mathrm{t}-1)d\sin(\theta_\mathrm{br})\sin(\varphi_{\mathrm{br}})$, $
\Gamma^\mathrm{br}_{i_z,i_y}= (i_z-1)d\cos(\theta_\mathrm{br})+(i_y-1)d\sin(\theta_\mathrm{br})\cos(\varphi_{\mathrm{br}})$,  with $\theta_\mathrm{br}$ and $\varphi_{\mathrm{br}}$ respectively being the azimuth and elevation angles of arrival (AoA) for signals from BS to FD-RIS, and $d$ is the space between elements of FD-RIS.
\begin{remark}
\normalfont{When a signal is incident on a RIS with periodic time modulation, the incident signal will be reflected into a series of harmonic signals \cite{zhang2018space, xiaoFD-RIS2024}. The frequency offsets between these harmonic signals and the central frequency are determined by the harmonic order and the modulation frequency. In practice, due to limited control abilities, it is challenging to precisely manipulate all harmonic signals, which may result in uncontrolled harmonic signals propagating in undesired directions. This can adversely affect the transmission capabilities and efficiency of other communication systems. Furthermore, high-order harmonic signals can spread into unnecessary frequency bands, occupying additional bandwidth resources and causing spectral pollution. To address these challenges, it is essential to identify a proficient modulation technique to transform the incident single into a single low-order harmonic signal.}
\end{remark}

To ensure that each subarray reflect the pure low-order harmonic signal, the periodic time modulation method presented in \cite{dai2020high} is leveraged. In this case, the $(i_z, i_y)$-th element's reflected coefficient of the FD-RIS can be expressed as
\begin{align}
&	\Upsilon_{i_z, i_y}(t)=A_0e^{j\phi_{i_z, i_y}(t)}, \label{eq_reflection_modulation}\\
& \phi_{i_z, i_y}(t)=\phi_0+ P\operatorname{mod}(t, T_{l}),
\end{align}
where index $l=	(\lceil i_z/M\rceil-1)S+\lceil i_y/N\rceil$, $P$ is the phase slope, $\operatorname{mod}(t, T_{l})$ represents the operation of taking the remainder with $t$ and $T_{l}$ representing the dividend and divisor, respectively. $A_0$ and $\phi_0$ represent the initial amplitude and phase of each element, respectively. Based on the theory of Fourier series expansion for periodic functions, the expression $\Upsilon_{i_z, i_y}(t)$ can be reformulated as
\begin{align}\label{eq_R_F}
	&\Upsilon_{i_z, i_y}(t)=\sum_{z=-\infty}^{z=+\infty}a^{i_z, i_y}_ze^{j2\pi z {f}_{l}t},
\end{align}
where $z\in\mathbb{Z}$ denotes the harmonic order number, ${f}_{l}$ is the modulation frequency of the $l$-th subarray given as ${f}_{l}=\frac{1}{T_l}$, $a^{t_z, t_y}_z$ represents the Fourier series coefficient given as
\begin{align}
	a^{i_z, i_y}_z=&\frac{1}{T_{l}}\int_{0}^{T_{l}}A_0e^{j(\phi_0+Pt)}e^{-j2\pi z {f}_{l}t}dt\notag\\
	=&\frac{jA_0e^{j\phi_0}\left(1-e^{j(PT_{l}-2z\pi)}\right)}{PT_{l}-2z\pi}.
\end{align}
It is worth noting that when $PT_{l}=2g\pi, ~g\in\mathbb{Z}$, $a^{i_z, i_y}_z$ can be further derived as
\begin{align}\label{eq_F_coeffcient}
	a^{i_z, i_y}_z=\begin{cases}
		A_0e^{j\phi_0}, &z=g,\\
		0,& z\neq g,
	\end{cases}
\end{align}
which means that we can obtain any single harmonic signal by appropriately selecting $P$. In other words, if $PT_l=2g\pi, ~g\in\mathbb{Z}$, the signal reflected by the $l$-th subarray contain only the $g$-th order harmonic signal. Substituting \eqref{eq_F_coeffcient} into \eqref{eq_R_F}, we have
\begin{align}
	\Upsilon_{i_z, i_y}(t)=A_0e^{j(2\pi g {f}_{l}t+\phi_0)}.
\end{align}
Hence, $x(t)$ can be re-expressed as
\begin{align}
	x(t)=&A_0e^{j\phi_0}\zeta(d_\mathrm{br}) s\sum_{n_\mathrm{t}=1}^{N_\mathrm{t}}w_{\mathrm{n_\mathrm
				t}}\sum_{i_z=1}^{I_z}\notag\\&\sum_{i_y=1}^{I_y}e^{j2\pi(f_\mathrm{c}+g{f}_{l})t}
		 e^{-\frac{j2\pi d^\mathrm{br}_{n_\mathrm{t}, i_z , i_y}}{\lambda}},
\end{align}
where $\lambda$ is the wave length of the carrier frequency.
\begin{remark}
\normalfont{	It is noteworthy that when employing the time modulation scheme as depicted in \eqref{eq_reflection_modulation}, the reflectivity coefficients assigned to the $g$-th order harmonic signal by all elements within each subarray become identical, i.e., $A_0e^{j(2\pi g {f}_{l}t+\phi_0)}$, which lacks of phase gradients in the reflectivity coefficients. This makes it challenging to spatially manipulate the $g$-th order harmonic signal towards specific directions according to the generalized Snell's law \cite{cui2014coding}.}
\end{remark}
To overcome this challenge, we introduce a time delay \textcolor{blue}{\footnote{\textcolor{blue}{The time delay represent a controllable temporal shift applied to the phase-control sequence of each RIS element. Such a time-shifted control process can be readily realized by adjusting the timing of the digital control signals generated by the FPGA.}}}, $\tau_i, i\in\mathscr{I}\triangleq\{1, \cdots, I\}$, into the periodical reflection coefficient of the $i$-th element in the FD-RIS. Specifically, according to the time-delay feature of the Fourier series, the reflection coefficient of the $i$-th element with a time delay $\tau_i$ can be expressed as
\begin{align}
	\Upsilon_{i_z, i_y}(t-\tau_i)=A_0e^{j\left(2\pi g {f}_{l}t+\phi_0-2\pi g {f}_{l}\tau_i\right)},
\end{align}
where $i_z=\lceil \frac{i}{I_y}\rceil$, $i_y=\operatorname{mod}(i,I_y)$, $i\in\mathscr{I}$.
The above conclusion indicates that the $i$-th element can offer an arbitrary phase-shift for the reflected harmonic signal by properly and carefully introducing a time delay. Consequently, implementing time delays on each element of the FD-RIS will significantly enhance the spatial modulation capabilities of the FD-RIS for the incident signals. In this case, the superimposition form of the reflected signals of the FD-RIS can be further expressed as
\begin{align}
	x(t)=&A_0e^{j\phi_0}\zeta(d_\mathrm{br}) s\sum_{n_\mathrm{t}=1}^{N_\mathrm{t}}w_{\mathrm{n_\mathrm
			t}}\sum_{i_z=1}^{I_z}\notag\\&\sum_{i_y=1}^{I_y}e^{j2\pi(f_\mathrm{c}+g{f}_{l})t}
e^{-\frac{j2\pi d^\mathrm{br}_{n_\mathrm{t}, i_z , i_y}}{\lambda}}e^{-j2\pi g{f}_l\tau_i}.
\end{align}
Therefore, the received signal at the position $(\hat{d}, \theta, \varphi)$ can be expressed as
\begin{align}\label{eq_received_signal}
	y(t)=\zeta(\hat{d})x\left(t-\frac{d_{i_z, i_y}}{c}\right)+n,
\end{align}
where $d_{i_z, i_y}=\hat{d}+\Gamma_{i_z, i_y}$ with $\Gamma_{i_z, i_y}=(i_z-1)d\cos(\theta)+(i_y-1)d\sin(\theta)\cos(\varphi)$, and $n\sim(0, \sigma^2)$ is the additive white Gaussian noise (AWGN) with $\sigma^2$ being the noise power. 

\textcolor{blue}{Next, the distance-angle beamforming of the proposed multi-subarray FD-RIS is verified. Specifically, we consider the case with $N_\mathrm{t}=1$ and temporarily ignore the large-scale path loss. Based on the received signal expression in \eqref{eq_received_signal}, the normalized beamforming gain can be derived as
	\begin{align}
		BF=&\frac{1}{I^2}\Bigg|\sum_{i_z=1}^{I_z}\sum_{i_y=1}^{I_y}e^{j2\pi g{f}_{l}(t-\frac{d_{iz, iy}}{c})}\times\notag\\
			&e^{-j2\pi g{f}_{l}\tau_i}e^{-\frac{j2\pi \left(\Gamma^\mathrm{br}_{i_z , i_y}+\Gamma_{i_z, i_y}\right)}{\lambda}}\Bigg|^2.
	\end{align} 
Then, we aim at maximizing the normalized beamforming gain $BF$ by appropriately selecting time delays, $\tau_i, i\in\mathscr{I}$, and time-modulation frequencies $f_l, l\in\mathcal{L}$. Fig. \ref{fig:pattern_verification} shows the optimal normalized beampattern of FD-RIS and traditional RIS. It is observed that conventional RIS can manipulate signal beamforming only in the angular domain. However, this single-dimensional capability becomes severely limited when multiple users are located at the same angle, as the system cannot effectively distinguish between them, thereby degrading their QoS. In contrast, FD-RIS enables simultaneous control in both the distance and angular domains, offering new opportunities to address this limitation. By leveraging dual-domain control, even when users share the same angular position, the system can exploit differences in the distance domain to reshape signal propagation. With this flexible spatial manipulation, FD-RIS shows great potential to overcome the beamforming limitations of conventional RIS.
\begin{figure}[ht]
	\centering
	\includegraphics[scale=0.45]{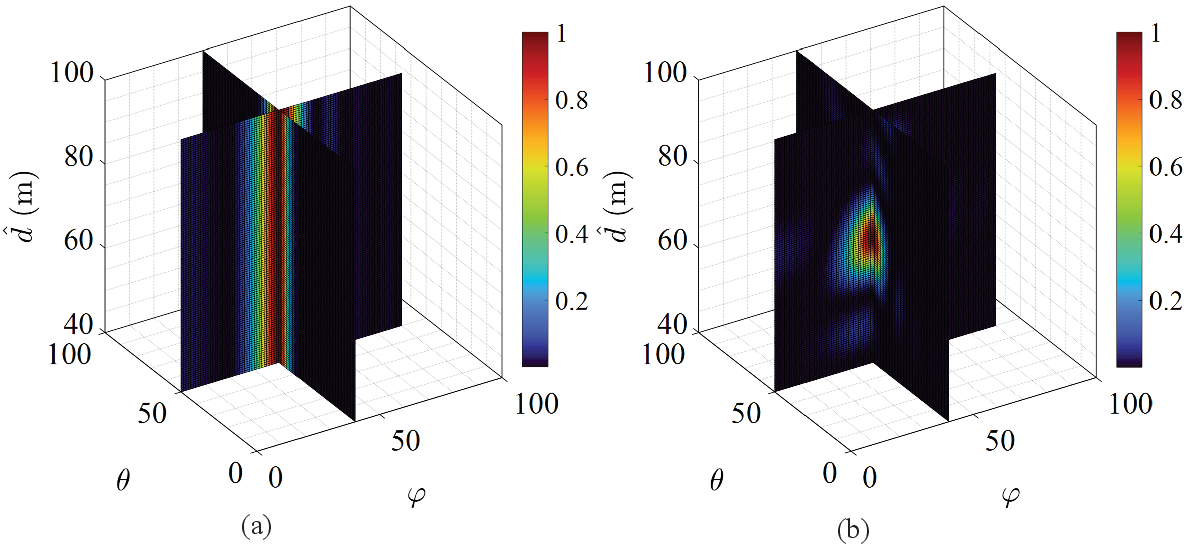}\\
	\caption{Optimal normalized beampattern considering $(\hat{d}, \theta, \varphi)=(70m, 50^\circ, 40^\circ)$, $f_l=0.2+20(l-1)/(L-1)$ MHz, $l\in\mathcal{L}$: (a) Conventional RIS; (b) FD-RIS.}\label{fig:pattern_verification}
\end{figure}
}
\begin{figure}[ht]
	\centering
	\includegraphics[scale=0.3]{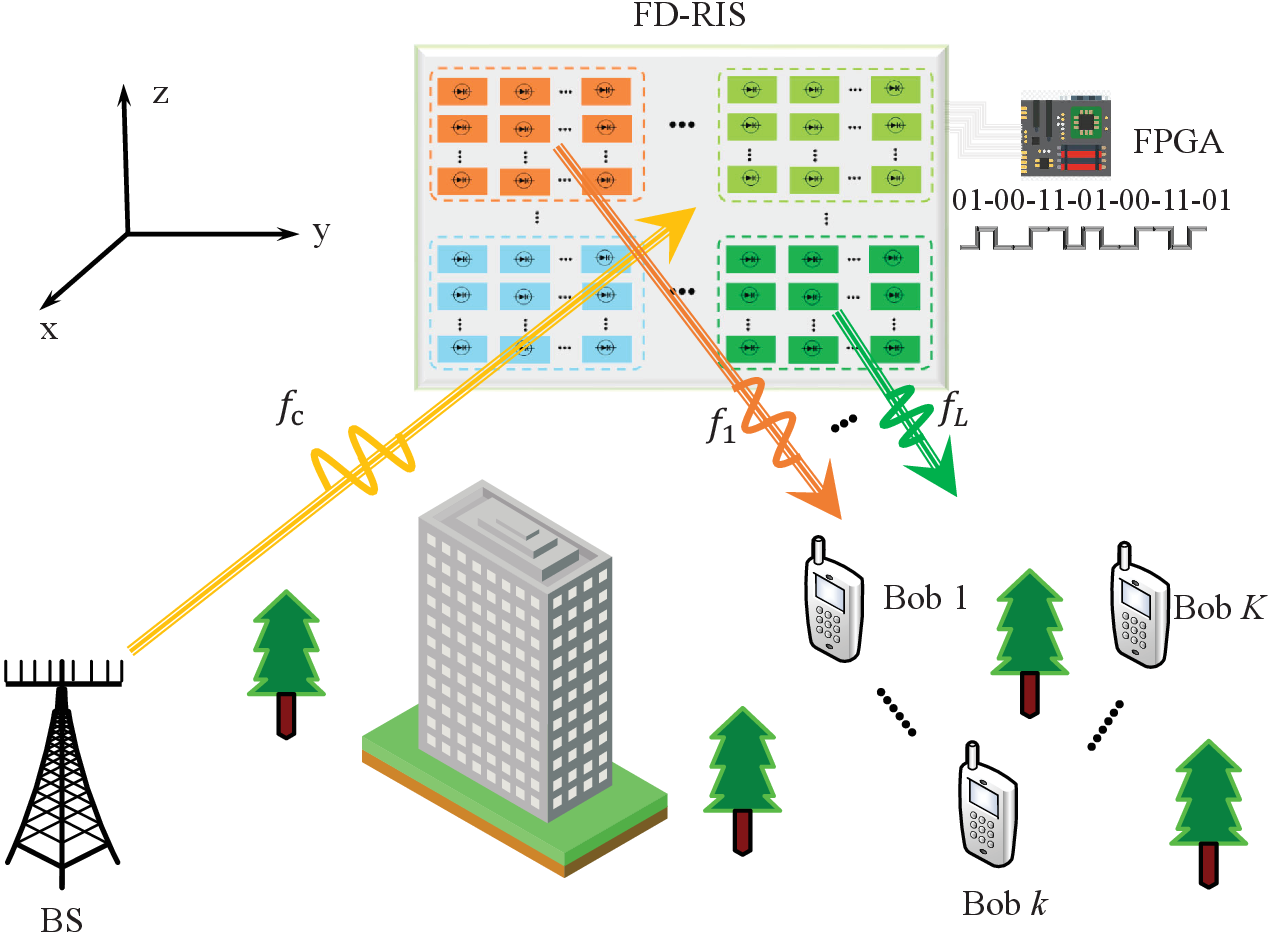}\\
	\caption{The multi-user wireless network supported by the multi-subarray FD-RIS.}\label{fig:scenario}
\end{figure}
\section{System Model and Problem Formulation}\label{sec:S3}
\subsection{System Model}
In this section, the efficacy of the multi-subarray FD-RIS is verified by considering the communication setup illustrated in Fig. \ref{fig:scenario}. This system comprises a BS equipped with \textcolor{blue}{a uniform linear array (ULA) of $N_\mathrm{t}$ antennas}, $K$ users each equipped with a single antenna, and a multi-subarray FD-RIS \textcolor{blue}{\footnote{\textcolor{blue}{The modulation frequencies applied to different subarrays do not alter the inherent response of the RIS. Because the RIS's responses for the incident signal depend solely on its carrier frequency. As long as the carrier frequency remains unchanged, the RIS responses under different states with respect to the incident signal, will not vary. It should be emphasized that the modulation frequencies in the RIS influence only the speed of phase-shift variation applied to the incident signal.}}} configured similarly to the FD-RIS described in section \ref{sec:S2}. It is assumed that the direct connection between the BS and users be blocked by buildings or trees, a situation frequently encountered in urban communication scenarios. Due to the existence of the LoS links between the BS and the FD-RIS, as well as between the FD-RIS and the users, the Rician channel model is utilized to represent the connections between the BS and the FD-RIS, i.e., $\mathbf{H}_\mathrm{BR}$, and also between the FD-RIS and the users, i.e.,  $\mathbf{h}_{\mathrm{r}k}$. In particular, channel $\mathbf{H}_\mathrm{BR}$ can be expressed as
\begin{align}
	\mathbf{H}_\mathrm{BR}=\zeta(d_\mathrm{br}) \left(\sqrt{\frac{\beta_1}{1+\beta}}\mathbf{H}_\mathrm{BR}^\mathrm{LoS}+\sqrt{\frac{1}{1+\beta_1}}\mathbf{H}_\mathrm{BR}^\mathrm{NLoS}\right),
\end{align}
where $\beta_1$ is the Rician factor. $\mathbf{H}_\mathrm{BR}^\mathrm{NLoS}$ denotes the non line-of-sight (NLoS) component, and each element of $\mathbf{H}_\mathrm{BR}^\mathrm{NLoS}$  is independent identically distributed (i.i.d.) and follows complex Gaussian distribution with zero mean and unit variance. According to the signal processing model of the FD-RIS in section \ref{sec:S2}, the LoS component, i.e., $\mathbf{H}_\mathrm{BR}^\mathrm{LoS}$, can be given as
\begin{align}
	\mathbf{H}_\mathrm{BR}^\mathrm{LoS}=\mathbf{a}_\mathrm{R}\mathbf{a}_\mathrm{B}^T,
\end{align}
where
\begin{itemize}
\item \textcolor{blue}{$\mathbf{a}_\mathrm{R}=\left[1, \cdots, e^{-j2\pi\frac{\Gamma^\mathrm{br}_{i_z, i_y}}{\lambda}}, \cdots, e^{-j2\pi\frac{\Gamma^\mathrm{br}_{I_z, I_y}}{\lambda}}\right]^T$ denotes the steering vector of FD-RIS, where $\Gamma^\mathrm{br}_{i_z,i_y}= (i_z-1)d\cos(\theta_\mathrm{br})+(i_y-1)d\sin(\theta_\mathrm{br})\cos(\varphi_{\mathrm{br}})$.}
	\vspace{2mm}
	\item \textcolor{blue}{$\mathbf{a}_\mathrm{B}=\Big[1,\cdots, e^{-j2\pi\frac{(n_\mathrm{t}-1)d\sin(\theta_\mathrm{br})\sin(\varphi_{\mathrm{br}})}{\lambda}}, \cdots,\\ e^{-j2\pi\frac{(N_\mathrm{t}-1)d\sin(\theta_\mathrm{br})\sin(\varphi_{\mathrm{br}})}{\lambda}}\Big]^T$ is the steering vector of the BS.}
\end{itemize}

Similarly, the link between the FD-RIS and the $k$-th user, i.e., $\mathbf{h}_{\mathrm{r}k}$, can be expressed as
\begin{align}
		\mathbf{h}_{\mathrm{r}k}=\zeta(d_{\mathrm{r}k}) \left(\sqrt{\frac{\beta_2}{1+\beta}}\mathbf{h}_{\mathrm{r}k}^\mathrm{LoS}+\sqrt{\frac{1}{1+\beta_2}}\mathbf{h}_{\mathrm{r}k}^\mathrm{NLoS}\right),
\end{align}
where $\beta_2$ is the Rician factor of FD-RIS-users link, $\zeta(d_{\mathrm{r}k})$ denotes the large-scale path loss with $d_{\mathrm{r}k}$ being the distance between the FD-RIS and the $k$-th user. $\mathbf{h}_{\mathrm{r}k}^\mathrm{NLoS}$ is the NLoS component, and each element of $\mathbf{h}_{\mathrm{r}k}^\mathrm{NLoS}$ is i.i.d. and follows complex Gaussian distribution with zero mean and unit variance. $\mathbf{h}_{\mathrm{r}k}^\mathrm{LoS}$ represents the LoS component, which is given by
\begin{align}
	\mathbf{h}_{\mathrm{r}k}^\mathrm{LoS}=\left[1, \cdots, e^{-j2\pi\frac{\Gamma^{\mathrm{r}k}_{i_z, i_y}}{\lambda}}, \cdots, e^{-j2\pi\frac{\Gamma^{\mathrm{r}k}_{I_z, I_y}}{\lambda}}\right]^T,
\end{align}
with $\Gamma^{\mathrm{r}k}_{i_z, i_y}=(i_z-1)d\cos(\theta_{\mathrm{r}k})+(i_y-1)d\sin(\theta_{\mathrm{r}k})\cos(\varphi_{\mathrm{r}k})$, where $\theta_{\mathrm{r}k}$ and $\varphi_{\mathrm{r}k}$ are the azimuth and elevation angles of departure (AoD) for signals from FD-RIS to  user $k$.

Based on the signal receiving model in section \ref{sec:S2}, the signal received by the $k$-th user can be expressed in a compact form as
\begin{align}	y_k=\mathbf{h}_{\mathrm{r}k}^T\tilde{\boldsymbol{\Theta}}\boldsymbol{\Theta}_k\mathbf{H}_{\mathrm{BR}}\mathbf{w}_ks_k+\sum_{j\neq k}^{K}\mathbf{h}_{\mathrm{r}k}^T\tilde{\boldsymbol{\Theta}}\boldsymbol{\Theta}_k\mathbf{H}_{\mathrm{BR}}\mathbf{w}_js_j+n_k,
\end{align}
where $n_k\sim\mathcal{CN}(0, \sigma_k^2)$ denotes the $k$-th user's AWGN with $\sigma_k^2$ being the noise power. In addition,
\begin{itemize}
	\item $\mathbf{w}_k=[w_k^1, \cdots, w_k^\mathrm{n_\mathrm{t}}, \cdots, w_k^\mathrm{N_\mathrm{t}}]^T,~ k\in\mathcal{K}\triangleq\{1, \cdots, K\},$ denotes the BS's active beamforming for the $k$-th user.
	\vspace{2mm}
	\item  $\boldsymbol{\Theta}_k=\operatorname{Diag}\{\theta^k_1, \cdots, \theta^k_i, \cdots, \theta^k_I\}, ~ k\in\mathcal{K}$  is the passive beamforming matrix associated with the frequency offsets and user $k$'s distance with  $\theta^k_i=\Upsilon_{i_z, i_y}(t)e^{-j2\pi g{f}_l\frac{d^{\mathrm{r}k}_{i_z, i_y}}{c}}$, which is unique to the FD-RIS.
	\vspace{2mm}
	\item $\tilde{\boldsymbol{\Theta}}=\operatorname{Diag}\{\tilde{\theta}_1, \cdots, \tilde{\theta}_i \cdots, \tilde{\theta}_I\}$ also represents the passive beamforming matrix with $\tilde{\theta}_i=e^{-j2\pi g{f}_l\tau_i}$, which is introduced by the time delays.
\end{itemize}

Thus, the achievable communication rate of the $k$-th user can be derived as
\begin{align}
	R_k=\log_2\Bigg(1+\frac{\left|\mathbf{h}_{\mathrm{r}k}^T\tilde{\boldsymbol{\Theta}}\boldsymbol{\Theta}_k\mathbf{H}_{\mathrm{BR}}\mathbf{w}_k\right|^2}{\sum_{j\neq k}^{K}\left|\mathbf{h}_{\mathrm{r}k}^T\tilde{\boldsymbol{\Theta}}\boldsymbol{\Theta}_k\mathbf{H}_{\mathrm{BR}}\mathbf{w}_j\right|^2+\sigma_k^2}\Bigg).	
\end{align}
\subsection{Problem Formulation}
To evaluate the effectiveness of the proposed multi-subarray FD-RIS, we formulate an optimization problem with the aim of maximizing $K$ users' weighted sum rate by jointly designing the active beamforming $\{\mathbf{w}_k\}_{k=1}^K$, the modulation frequencies $\mathbf{f}=[{f}_1, \cdots, {f}_L]^T$ and time delays $\boldsymbol{\tau}=[\tau_1, \cdots, \tau_I]^T$. The established optimization problem is given by
\begin{subequations}\label{eq_orig_opti}
	\begin{align}
		&\max _{\{\mathbf{w}_k\}_{k=1}^K, \mathbf{f}, \boldsymbol{\tau}}~ \sum_{k=1}^{K} \omega_k R_k,\notag \\
		&\qquad\text { s.t. }~f_{\min}\leq {f}_l\leq f_{\max},~\forall l\in\mathcal{L}, \label{eq_orig_opti_1}\\
		&\qquad\qquad~\sum_{k=1}^{K}\left\|\mathbf{w}_k\right\|^2\leq P_\mathrm{tmax}, \label{eq_orig_opti_2}\\
		&\qquad\qquad~\tau_i\geq 0,~\forall i\in\mathscr{I},\label{eq_orig_opti_3}
	\end{align}
\end{subequations}
where $\omega_k, ~k\in\mathcal{K},$ is the weighted factor of the $k$-th user, which can be utilized to control the priority of the $k$-th user and the fairness among the users. $f_{\min}$ and $f_{\max}$ denote the lower and upper bounds of the modulation frequency ${f}_l,$ for $l\in\mathcal{L}$. $P_\mathrm{tmax}$ is the maximum power budget of the BS. Actually, it is challenging to address this optimization problem, since the non-convexity of the objective function with respect to (w.r.t.) $\{\mathbf{w}_k\}_{k=1}^K$, $\mathbf{f}$ and $ \boldsymbol{\tau}$.
To handle this problem, the MMSE method \cite{shi2011iteratively} is first utilized to transform the objective function expressed in logarithmic form. Then,
the alternating strategy is leveraged to divide the transformed optimization problem into three subproblems, which are respectively for the design of active beamforming, time-delay variables, and modulation frequency variables. Specifically, leveraging the MMSE method and introducing auxiliary variables $W_k$ and $\mu_k$, the communication rate $R_k, ~k\in\mathcal{K},$ can be converted as
\begin{align}\label{eq_mmse_w}
	R_k=&\frac{1}{\ln2}\max_{W_k, u_k\geq 0}\ln (W_k)-\notag\\
	&W_kE_k(u_k, \{\mathbf{w}_k\}_{k=1}^K, \mathbf{f}, \boldsymbol{\tau})+1,
\end{align}
where
$E_k=u_k^H\Big(\sum_{j\neq k}^{K} \left|\mathbf{h}_{\mathrm{r}k}^T\tilde{\boldsymbol{\Theta}}\boldsymbol{\Theta}_k\mathbf{H}_{\mathrm{BR}}\mathbf{w}_j\right|^2+\sigma_k^2\Big)u_k+\left(u_k^H\mathbf{h}_{\mathrm{r}k}^T\tilde{\boldsymbol{\Theta}}\boldsymbol{\Theta}_k\mathbf{H}_{\mathrm{BR}}\mathbf{w}_k-1\right)\left(u_k^H\mathbf{h}_{\mathrm{r}k}^T\tilde{\boldsymbol{\Theta}}\boldsymbol{\Theta}_k\mathbf{H}_{\mathrm{BR}}\mathbf{w}_k-1\right)^H$.
  It is worth noting that the optimal $W_k$ and $u_k$ can be derived easily by setting the first-order partial derivative of the objective functions in \eqref{eq_mmse_w}  w.r.t. $W_k$ and $u_k$ to zero. Specifically, the optimal solutions can be given by
\begin{align}\label{eq_opt_W_k_u_k}
	\begin{cases}
		W_k^\mathrm{opt}=E_k^{-1},&\\ u_k^\mathrm{opt}=\frac{\mathbf{h}_{\mathrm{r}k}^T\tilde{\boldsymbol{\Theta}}\boldsymbol{\Theta}_k\mathbf{H}_{\mathrm{BR}}\mathbf{w}_k}{\sum_{j=1}^{K}\left|\mathbf{h}_{\mathrm{r}k}^T\tilde{\boldsymbol{\Theta}}\boldsymbol{\Theta}_k\mathbf{H}_{\mathrm{BR}}\mathbf{w}_j\right|^2+\sigma_k^2}.&
	\end{cases}
\end{align}

The analysis above indicates that the lower bound of  $R_k$ in the $(q+1)$-th iteration can be obtained with the given $W_k^{(q)}$ and $u_k^{(q)}$. Note that the values of $W_k^{(q)}$ and $u_k^{(q)}$ are calculated by utilizing the optimal expressions in \eqref{eq_opt_W_k_u_k} with the achieved active beamforming, time delays and modulation frequency solutions in the $q$-th iteration. Hence, the lower bound of  $R_k$ in the $(q+1)$-th iteration can be expressed as
$
	R_k\geq\tilde{R}_k=\frac{\ln (W_k^{(q)})-W_k^{(q)}E_k^{(q)}+1}{\ln2}.
$
Thus, in $(q+1)$-th iteration, the original optimization problem can be transformed as
	\begin{align}\label{eq_orig_opti_mmse}
		&\max _{\{\mathbf{w}_k\}_{k=1}^K, \mathbf{f}, \boldsymbol{\tau}}~ \sum_{k=1}^{K} \omega_k \tilde{R}_k,\notag \\
		&\qquad\text { s.t. }~\eqref{eq_orig_opti_1}-\eqref{eq_orig_opti_3}.
	\end{align}
Next, we will utilize the alternating strategy to divide the above problem into three subproblems pertaining to the design of the active beamforming, time delay and modulation frequency in the next section.

\section{Algorithm Design}\label{sec:S4}
\subsection{Active beamforming design }This section focus on optimizing the active beamforming variables $\{\mathbf{w}_k\}_{k=1}^K$ with the given time-delay variable $\boldsymbol{\tau}$ and modulation frequency variable $\mathbf{f}$. In this case, the optimization problem \eqref{eq_orig_opti_mmse} can be simplified as the following subproblem
	\begin{align}\label{eq_active}
		&\max _{\{\mathbf{w}_k\}_{k=1}^K} \sum_{k=1}^{K}\omega_k\tilde{R}_k(\{\mathbf{w}_k\}_{k=1}^K),\notag \\
		&\quad\text { s.t. }~\sum_{k=1}^{K}\left\|\mathbf{w}_k\right\|^2\leq P_\mathrm{tmax}.
	\end{align}
It is important to highlight that the optimization problem \eqref{eq_active} is a convex quadratic constrained quadratic programming (QCQP) problem,  and it can be effectively solved by the existing convex optimization tools such as CVX. However, we have to face the high computational complexity challenge. To address this problem with a low complexity, we resort to the Lagrange multiplier method. Specifically, the Lagrange function of the problem \eqref{eq_active} can be expressed as
\begin{align}\label{eq_lagrange}
	\mathscr{L}=&-\sum_{k=1}^{K}\omega_k\tilde{R}_k(\{\mathbf{w}_k\}_{k=1}^K)+\mu\left(\sum_{k=1}^{K}\left|\mathbf{w}_k\right|^2-P_\mathrm{tmax}\right),\notag\\
	=&\sum_{k=1}^{K}\omega_k\Bigg[\frac{W_k^{(q)}}{\ln2}\Bigg(\sum_{j=1}^{K}\mathbf{w}_j^H\mathbf{A}_k\mathbf{w}_j-2\operatorname{Re}(\mathbf{a}_k^H\mathbf{w}_k)\notag\\
	&+\sigma_k^2\left(u_k^{(q)}\right)^Hu_k^{(q)}+1\Bigg)-\frac{\ln \left(W_k^{(q)}\right)+1}{\ln2}\Bigg]\notag\\
	&+\mu\left(\sum_{k=1}^{K}\left\|\mathbf{w}_k\right\|^2-P_\mathrm{tmax}\right),
\end{align}
where $\mu\geq 0$ is the Lagrange multiplier (dual variable) and
\begin{itemize}
	\item $\mathbf{A}_k=\mathbf{H}_{\mathrm{BR}}^H\boldsymbol{\Theta}_k^H\tilde{\boldsymbol{\Theta}}^H\mathbf{h}_{\mathrm{r}k}^*u_k^{(q)}\left(u_k^{(q)}\right)^H\mathbf{h}_{\mathrm{r}k}^T\tilde{\boldsymbol{\Theta}}\boldsymbol{\Theta}_k\mathbf{H}_{\mathrm{BR}}$.
\item $\mathbf{a}_k=\mathbf{H}_{\mathrm{BR}}^H\boldsymbol{\Theta}_k^H\tilde{\boldsymbol{\Theta}}^H\mathbf{h}_{\mathrm{r}k}^*u_k^{(q)}$.
\end{itemize}
 Thus, the first-order partial derivative of the Lagrange function w.r.t. $\mathbf{w}_j$ can be derived as
\begin{align}
\frac{\partial \mathscr{L}}{\partial \mathbf{w}_j}=2\sum_{k=1}^{K}\frac{\omega_kW_k^{(q)}\mathbf{A}_k}{\ln2}\mathbf{w}_j-2\frac{\omega_jW_j^{(q)}\mathbf{a}_j}{\ln2}+2\mu\mathbf{w}_j.
\end{align}
Let $\frac{\partial \mathscr{L}}{\partial \mathbf{w}_j}=0$, we can derive the optimal $\mathbf{w}_j$ associated with the dual variable $\mu$, which is given by
\begin{align}\label{eq_w_k_opt_dual}
	\mathbf{w}^\mathrm{opt}_j(\mu)=\left(\widehat{\mathbf{A}}+\mu\mathbf{I}_{N_\mathrm{t}\times N_\mathrm{t}}\right)^{\dagger}\frac{\omega_jW_j^{(q)}\mathbf{a}_j}{\ln2},~ j\in\mathcal{K},
\end{align}
where $\widehat{\mathbf{A}}=\sum_{k=1}^{K}\frac{\omega_kW_k^{(q)}\mathbf{A}_k}{\ln2}$.
\begin{remark}
	\normalfont{Note that the feasible solution strictly satisfying the constraint $\sum_{k=1}^{K}\left|\mathbf{w}_k\right|^2- P_\mathrm{tmax}< 0$ can always be found. Thus,  the optimization problem \eqref{eq_active} meets the Slater's condition, which means that the dual gap between problem \eqref{eq_active} and its dual problem is zero. Consequently, we can achieve the optimal solution of problem \eqref{eq_active} by addressing its dual problem and obtaining the optimal dual variable \cite{pan2020intelligent}.}
\end{remark}

Let $\mu^\mathrm{opt}$ represent the optimal dual variable, which is required to satisfy the complementary slackness condition associated with the power constraint:
\begin{align}
	\mu^\mathrm{opt}\left(\sum_{k=1}^{K}\left\|\mathbf{w}^\mathrm{opt}_k(\mu^\mathrm{opt})\right\|^2-P_\mathrm{tmax}\right)=0.
\end{align}
In particular, if the following condition holds:
\begin{align}\label{eq_comple_slack}
\sum_{k=1}^{K}\left\|\mathbf{w}^\mathrm{opt}_k(\mu^\mathrm{opt})\right\|^2-P_\mathrm{tmax}<0,
\end{align}
we have $\mu^\mathrm{opt}=0$ and the optimal solution of the problem \eqref{eq_active} is $\mathbf{w}^\mathrm{opt}_k(0)=\widehat{\mathbf{A}}^{\dagger}\frac{\omega_kW_k^{(q)}\mathbf{a}_k}{\ln2},~ k\in\mathcal{K}$. Otherwise we have
\begin{align}\label{eq_mu_P}
	P(\mu^\mathrm{opt})=\sum_{k=1}^{K}\left\|\mathbf{w}^\mathrm{opt}_k(\mu^\mathrm{opt})\right\|^2=P_\mathrm{tmax}.
\end{align}
Note that the intricate form of $\mathbf{w}^\mathrm{opt}_k$ in \eqref{eq_w_k_opt_dual} makes it difficult to obtain the solution of $\mu^\mathrm{opt}$ directly through \eqref{eq_mu_P}. Although some numerical methods such as the bisection search method can be leveraged to achieve the $\mu^\mathrm{opt}$, understanding the monotonic nature of the function $P(\mu^\mathrm{opt})$ is essential. Actually, it is difficult to prove the monotonic nature of the function $P(\mu^\mathrm{opt})$ based on its complex explicit expression. Next, we will leverage the following theorem to prove the monotonicity of the function $P(\mu^\mathrm{opt})$.
\begin{theorem}\label{th1}
	The function $P(\mu^\mathrm{opt})$ is a monotonically decreasing function w.r.t. $\mu^\mathrm{opt}$.
\end{theorem}
\begin{proof}
		The proof is given in Appendix \ref{append 1}.
\end{proof}
\vspace{2mm}
According to the Theorem \ref{th1}, the bisection search method can be utilized to find the only solution $\mu^\mathrm{opt}$ that makes $P(\mu^\mathrm{opt})=P_\mathrm{tmax}$, and further achieve the optimal active beamforming. The step-by-step instructions for solving the optimization problem \eqref{eq_active} can be found in Algorithm 1.
 \begin{center}
	\begin{tabular}{p{8.5cm}}
		\toprule[2pt]
		\textbf{Algorithm 1:}  Bisection Search Method for Solving Active Beamforming Subproblem \eqref{eq_active}   \\
		\midrule[1pt]
		1: Initialize the tolerance accuracy $\epsilon$ and the bound of the\\
		\quad dual variable $\mu_\mathrm{up}$ and $\mu_\mathrm{low}$.\\
		2: If condition \eqref{eq_comple_slack} holds, $\mu^\mathrm{opt}=0$. Otherwise, go to steps\\\quad 3-6 to solve $\mu^\mathrm{opt}$ in \eqref{eq_mu_P}.\\
		3: \textbf{Repeat} \\
		4: \quad Calculate $\mu^\mathrm{opt}=\frac{\mu_\mathrm{up}+\mu_\mathrm{low}}{2}$ and $P(\mu^\mathrm{opt})$.\\
		5: \quad if $P(\mu^\mathrm{opt})\geq P_\mathrm{tmax}$, $ \mu_\mathrm{low}\leftarrow\mu^\mathrm{opt}$. Otherwise, $\mu_\mathrm{up}\leftarrow$\\\qquad $\mu^\mathrm{opt}$. \\
		6: \textbf{Until} $|\mu_\mathrm{up}-\mu_\mathrm{low}|\leq \epsilon$\\
		7: \textbf{Output:} the optimal dual variable $\mu^\mathrm{opt}$ and the optimal\\\quad active beamforming $\mathbf{w}^\mathrm{opt}_k(\mu^\mathrm{opt})$ through \eqref{eq_w_k_opt_dual}. \\
		\bottomrule[2pt]
	\end{tabular}
\end{center}
\subsection{Passive Beamforming Design for Time Delay }
The time delays are designed in this section with the obtained active beamforming variables and the given modulation frequency. Specifically, the subproblem associated with the time-delay variables can be expressed as
	\begin{align}\label{eq_time_delay}
		&\max _{\boldsymbol{\tau}}~~ \sum_{k=1}^{K} \omega_k\tilde{R}_k(\boldsymbol{\tau}),\notag \\
		&~\text { s.t. }~~\tau_i\geq 0, ~i\in\mathscr{I}.
	\end{align}
Note that it is difficult to address this particular subproblem due to the non-convex nature of $\tilde{R}_k(\boldsymbol{\tau})$ w.r.t. $\boldsymbol{\tau}$ and the coupling between  $\boldsymbol{\tau}$ and $\tilde{\theta}_i=e^{-j2\pi g{f}_l\tau_i}, ~ i\in\mathscr{I}$.
\begin{remark}\label{re_4}
	\normalfont{From the expression of $\tilde{\theta}_i=e^{-j2\pi g{f}_l\tau_i}, ~ i\in\mathscr{I}$, it is worth noting that regardless the value of ${f}_l, l\in\mathcal{L}$, each element in the FD-RIS has the capability to induce a specific phase-shift on the reflected harmonic signal by introducing an appropriate time delay. Therefore, we can initially design the reflection coefficients generated by the time delays in a holistic manner. \textcolor{blue}{In other words,  $\tilde{\theta}_i = e^{-j2\pi g f_l \tau_i}, ~ i\in\mathscr{I},$ can be treated as the optimization variable and designed, rather than focusing separately on the design of $\tau_i, ~ i\in\mathscr{I}$.
	 Subsequently, it is straightforward to derive the corresponding time delays $\tau_i, ~ i\in\mathscr{I}$ with given ${f}_l, ~l\in\mathcal{L}$ based on the obtained reflection coefficients $\tilde{\theta}_i,~ i\in\mathscr{I}$ in the final step of the iteration.}}
\end{remark}

Thus, the optimization problem \eqref{eq_time_delay} can be equivalently transformed as
	\begin{align}\label{eq_time_delay_trans}
		&\max _{\boldsymbol{\vartheta}}~~  f(\boldsymbol{\vartheta})=-\boldsymbol{\vartheta}^H\mathbf{B}\boldsymbol{\vartheta}+2\operatorname{Re}(\boldsymbol{\vartheta}^H\mathbf{b})+C,\notag \\
		&~\text { s.t. }~~|\tilde{\theta}_i^*|=1,~ \angle\tilde{\theta}^*_i\in[0, 2\pi],~\forall i\in\mathscr{I},
	\end{align}
considering
\begin{align}
	\sum_{k=1}^{K} \tilde{R}_k&=\sum_{k=1}^{K}\omega_k\Bigg(\frac{-{W}_k^{(q)}}{\ln2}\Big(\boldsymbol{\vartheta}^H\mathbf{B}_k\boldsymbol{\vartheta}-2\operatorname{Re}(\boldsymbol{\vartheta}^H\mathbf{b}_k)+\notag\\
	&1+\left({u}_k^{(q)}\right)^H{u}_k^{(q)}\sigma_k^2\Big)+\frac{\ln ({W}_k^{(q)})+1}{\ln2}\Bigg)\notag\\
	&=-\boldsymbol{\vartheta}^H\mathbf{B}\boldsymbol{\vartheta}+2\operatorname{Re}(\boldsymbol{\vartheta}^H\mathbf{b})+C,
\end{align}
where
\begin{itemize}
	\item  $\mathbf{B}_k\!=\!\left({u}_k^{(q)}\right)^H\mathbf{H}_{\mathrm{r}k}\boldsymbol{\Theta}_k\mathbf{H}_\mathrm{BR}\sum\limits_{j=1}^{K}\mathbf{w}_j\mathbf{w}_j^H\mathbf{H}_\mathrm{BR}^H\boldsymbol{\Theta}_k^H\mathbf{H}_{\mathrm{r}k}^H{u}_k^{(q)}$,
	\vspace{2mm}
	\item $\mathbf{b}_k=\left({u}_k^{(q)}\right)^H\mathbf{H}_{\mathrm{r}k}\boldsymbol{\Theta}_k\mathbf{H}_\mathrm{BR}\mathbf{w}_k$, $\mathbf{H}_{\mathrm{r}k}=\operatorname{Diag}(\mathbf{h}_{\mathrm{r}k})$,
	\vspace{2mm}
	\item $\boldsymbol{\vartheta}^*=\operatorname{diag}(\tilde{\boldsymbol{\Theta}})$, $\mathbf{B}=\sum\limits_{k=1}^{K}\frac{\omega_k{W}_k^{(q)}\mathbf{B}_k}{\ln2}$, $\mathbf{b}=\sum\limits_{k=1}^{K}\frac{\omega_k{W}_k^{(q)}\mathbf{b}_k}{\ln2}$,
	\vspace{2mm}
	\item
	$C=-\sum\limits_{k=1}^{K}\omega_k\frac{{W}_k^{(q)}\left(1+\left({u}_k^{(q)}\right)^H{u}_k^{(q)}\sigma_k^2\right)-\ln ({W}_k^{(q)})-1}{\ln2}$.
\end{itemize}
It is worth noting problem \eqref{eq_time_delay_trans} is still a non-convex optimization problem because of the non-convexity of the constant modulus constraints and the objective function.
 To tackle this issue, the Riemannian conjugate gradient algorithm \cite{alhujaili2019transmit} is leveraged to effectively solve this optimization problem with low computing complexity and fast convergence. Specifically, the feasible region of problem \eqref{eq_time_delay_trans} can be regarded as the product of $I$ complex circles, which is denoted as $\mathcal{S}\triangleq\{\boldsymbol{\vartheta}\in\mathbb{C}^{I\times 1}: |\tilde{\theta}^*_1|=\cdots=|\tilde{\theta}^*_i|=\cdots=|\tilde{\theta}^*_I|=1\}$. The core concept of the  Riemannian conjugate gradient algorithm can be summarized as follows: First, calculate the Riemannian gradient within the defined manifold space $\mathcal{S}$ to determine the direction of the steepest descent on the manifold, resembling the concept of gradient in the Euclidean space. Subsequently, update the next iteration point using the Riemannian gradient. As the Riemannian gradient lies on the tangent space of the manifold, denoted as $\mathcal{T}_{\boldsymbol{\vartheta}}\mathcal{S}$,  it is necessary to map the updated iteration point back to the manifold space.

The Riemannian gradient of $f(\boldsymbol{\vartheta})$ at point $\boldsymbol{\vartheta}^{(e)}$ in the $e$-th iteration can be derived by orthogonally mapping the gradient in Euclidean space to tangent space, which is given by \cite{absil2008optimization}
\begin{align}\label{eq_Riemannian_gradient}
	&\operatorname{grad}f(\boldsymbol{\vartheta}^{(e)})=\mathscr{M}_{\mathcal{T}_{\boldsymbol{\vartheta}^{(e)}}\mathcal{S}}\left(\nabla f\left(\boldsymbol{\vartheta}^{(e)}\right)\right)
\notag\\
&=\nabla f\left(\boldsymbol{\vartheta}^{(e)}\right)-\operatorname{Re}\left(\nabla f\left(\boldsymbol{\vartheta}^{(e)}\right)\circ \left(\boldsymbol{\vartheta}^{(e)}\right)^*\right)\circ \boldsymbol{\vartheta}^{(e)},
\end{align}
where $\nabla f(\boldsymbol{\vartheta}^{(e)})$ is the gradient of $f(\boldsymbol{\vartheta})$ in Euclidean space, which is expressed as
	$\nabla f(\boldsymbol{\vartheta}^{(e)})=-2\mathbf{B}\boldsymbol{\vartheta}^{(e)}+2\mathbf{b}.$
Thus, the search direction for the next iteration point can be given by
\begin{align}\label{eq_search_direction}
	\mathscr{D}^{(e)}=\operatorname{grad}f\left(\boldsymbol{\vartheta}^{(e)}\right)+\psi_\mathrm{PR}^{(e)}\mathscr{M}_{\mathcal{T}_{\boldsymbol{\vartheta}^{(e)}}\mathcal{S}}\left(\mathscr{D}^{(e-1)}\right),
\end{align}
where $\psi_\mathrm{PR}^{(e)}$ represents the Polak-Ribiere parameter \cite{shewchuk1994introduction} whose expression can be expressed as \eqref{eq_Polak-Ribiere}, shown at the top of this page. The operator $\langle\cdot, \cdot\rangle$ denotes the inner product operation.
\begin{figure*}[t]
	\begin{align}\label{eq_Polak-Ribiere}
		\psi_\mathrm{PR}^{(e)}=\mathcal{R}\left(\frac{\left\langle \operatorname{grad}f\left(\boldsymbol{\vartheta}^{(e)}\right), \operatorname{grad}f\left(\boldsymbol{\vartheta}^{(e)}\right)-\mathscr{M}_{\mathcal{T}_{\boldsymbol{\vartheta}^{(e)}}\mathcal{S}}\left(\operatorname{grad}f\left(\boldsymbol{\vartheta}^{(e-1)}\right)\right)\right\rangle}{\left\langle\operatorname{grad}f\left(\boldsymbol{\vartheta}^{(e-1)}\right),\operatorname{grad}f\left(\boldsymbol{\vartheta}^{(e-1)}\right)\right\rangle}\right).
	\end{align}
	\hrule
\end{figure*}

Hence, the updated rule of variable $\boldsymbol{\vartheta}$ can be expressed as
\begin{align}\label{eq_update_rule}
	\boldsymbol{\vartheta}^{(e+1)}=\mathscr{P}\Big(\boldsymbol{\vartheta}^{(e)}+\eta^{(e)}\mathscr{D}^{(e)}\Big),
\end{align}
where  $\mathscr{P}(\cdot)$ is the projection operator, which maps the updated $\boldsymbol{\vartheta}$ to the feasible region $\mathcal{S}$ by element-wise retraction, i.e., $\frac{(\boldsymbol{\vartheta})_i}{|(\boldsymbol{\vartheta})_i|}, ~i\in\mathscr{I}$. $\eta^{(e)}$ denotes the Armijo step in the $(e+1)$-th iteration.
To obtain a suitable step size, the Armijo criteria \cite{shewchuk1994introduction} is usually leveraged, which is given by
\begin{align}\label{eq_Armijo_criteria}
	f\left(\mathscr{P}\left(\boldsymbol{\vartheta}^{(e)}+\eta^{(e)}\mathscr{D}^{(e)}\right)\right)&\geq f\left(\boldsymbol{\vartheta}^{(e)}\right)+\notag\\
	&c_1\eta^{(e)}\operatorname{grad}f\left(\vartheta^{(e)}\right)^H\mathscr{D}^{(e)},
\end{align}
where $c_1\in(0, 1)$ is a parameter that controls the amount of decrease in the objective function.
The details for solving the optimization problem \eqref{eq_time_delay_trans} utilizing the  Riemannian conjugate gradient algorithm are presented in Algorithm 2.
\begin{center}
	\begin{tabular}{p{8.5cm}}
		\toprule[2pt]
		\textbf{Algorithm 2:}  Riemannian Conjugate Gradient Algorithm for Solving Problem \eqref{eq_time_delay_trans}   \\
		\midrule[1pt]
		1: Initialize $\boldsymbol{\vartheta}^{(0)}$, tolerance accuracy $\tilde{\epsilon}$ and step size $\eta^{(0)}$;\\\quad Set iteration index $e=0$.\\
		2: \textbf{Repeat} \\
		3: \quad Calculate objective function $f(\boldsymbol{\vartheta}^{(e)})$ and Riemannian\\\qquad gradient $\operatorname{grad}f(\boldsymbol{\vartheta}^{(e)})$ utilizing \eqref{eq_Riemannian_gradient}.\\
		4: \quad Calculate search direction $\mathscr{D}^{(e)}$ and Polak-Ribiere\\\qquad parameter $\psi_\mathrm{PR}^{(e)}$ using \eqref{eq_search_direction} and \eqref{eq_Polak-Ribiere}, respectively.\\
		5: \quad Search the appropriate step size utilizing the Armijo \\\qquad criteria, update $\boldsymbol{\vartheta}^{(e+1)}$ according to \eqref{eq_update_rule} and calculate\\\qquad the corresponding objective function $f(\boldsymbol{\vartheta}^{(e+1)})$;\\\qquad Let $e\leftarrow e+1$. \\
		6: \textbf{Until} $|f(\boldsymbol{\vartheta}^{(e+1)})-f(\boldsymbol{\vartheta}^{(e)})|\leq \tilde{\epsilon}$.\\
		7: \textbf{Output:} the optimal $\boldsymbol{\vartheta}^\mathrm{opt}$. \\
		\bottomrule[2pt]
	\end{tabular}
\end{center}
\subsection{Passive Beamforming Design for Modulation Frequency}
In this section, we aim to optimize the modulation frequency $\mathbf{f}$ with the obtained values of  $\mathbf{w}_k, k\in\mathcal{K}$, and $\boldsymbol{\tau}$. The subproblem concerning the modulation frequency variables is described as:
	\begin{align}\label{eq_modulation_frequency}
		&\min _{\mathbf{f}}~~ \tilde{g}(\mathbf{f})=\sum_{k=1}^{K}\boldsymbol{\vartheta}_k^H\mathbf{D}_k\boldsymbol{\vartheta}_k-2\operatorname{Re}(\boldsymbol{\vartheta}_k^H\mathbf{d}_k)-\hat{C},\notag \\
		&~\text { s.t. }~f_{\min}\leq {f}_l\leq f_{\max},~\forall l\in\mathcal{L},
	\end{align}
taking account of
\begin{align}
	\sum_{k=1}^{K}\omega_k\tilde{R}_k=\sum_{k=1}^{K}-\boldsymbol{\vartheta}_k^H\mathbf{D}_k\boldsymbol{\vartheta}_k+2\operatorname{Re}(\boldsymbol{\vartheta}_k^H\mathbf{d}_k)+\hat{C}
\end{align}
where
\begin{itemize}
	\item  $\mathbf{D}_k=\frac{\omega_k{W}_k^{(q)}}{\ln2}\left({u}_k^{(q)}\right)^H\mathbf{H}_{\mathrm{r}k}\tilde{\boldsymbol{\Theta}}\mathbf{H}_\mathrm{BR}\sum\limits_{j=1}^{K}\mathbf{w}_j\mathbf{w}_j^H\times\\
	\mathbf{H}_\mathrm{BR}^H\tilde{\boldsymbol{\Theta}}^H\mathbf{H}_{\mathrm{r}k}^H{u}_k^{(q)}$,
	\vspace{2mm}
	\item $\mathbf{d}_k=\frac{\omega_k{W}_k^{(q)}}{\ln2}\left({u}_k^{(q)}\right)^H\mathbf{H}_{\mathrm{r}k}\tilde{\boldsymbol{\Theta}}\mathbf{H}_\mathrm{BR}\mathbf{w}_k$,
	\vspace{2mm}
	\item $\mathbf{H}_{\mathrm{r}k}=\operatorname{Diag}(\mathbf{h}_{\mathrm{r}k})$, $\boldsymbol{\vartheta}_k^*=\operatorname{diag}(\boldsymbol{\Theta}_k)$,
	\vspace{2mm}
	\item
	$\hat{C}=-\sum\limits_{k=1}^{K}\omega_k\frac{{W}_k^{(q)}\left(1+\left({u}_k^{(q)}\right)^H{u}_k^{(q)}\sigma_k^2\right)-\ln ({W}_k^{(q)})-1}{\ln2}$.
\end{itemize}
Although the objective function $\tilde{g}$ is convex w.r.t. $\boldsymbol{\vartheta}_k$, the interconnection between  ${f}_l, ~l\in\mathcal{L},$ and $\theta^k_i=e^{-j2\pi g{f}_l\frac{d_{i_z, i_y}^{\mathrm{r}k}}{c}}, ~k\in\mathcal{K}, i\in\mathscr{I}$ makes the objective function $\tilde{g}$  non-convex concerning $\mathbf{f}$.
Thus, it is difficult to handle this optimization problem directly.
To effectively tackle this challenge, the GCMMA algorithm  \cite{svanberg2002class, svanberg2007mma} is adopted. The core principle of the GCMMA algorithm is to utilize the first-order gradients of the objective function and constraints w.r.t. the optimized variables to construct the MMA convex approximation subproblem, which can approximate the original problem. Specifically, in the $z$-th iteration, the MMA convex approximation subproblem can be given by
	\begin{align}\label{eq_modulation_frequency_GCMMA}
		&\min _{\mathbf{f}}~~ \widehat{g}(\mathbf{f}|\mathbf{f}^{(z-1)}),\notag \\
		&~\text { s.t. }~f_{l, \mathrm{low}}^{(z)}\leq {f}_l\leq f_{l, \mathrm{up}}^{(z)},~\forall l\in\mathcal{L},
	\end{align}
where $\widehat{g}(\mathbf{f}|\mathbf{f}^{(z-1)})$ denotes the MMA convex approximation of the objective function, which is given by
\begin{align}\label{eq_MMA_approxi}
	\widehat{g}(\mathbf{f}|\mathbf{f}^{(z-1)})=&\sum_{l=1}^{L} \frac{p^{(z)}_l}{u_l^{(z)}-{f}_l}+\frac{q^{(z)}_l}{{f}_l-o_l^{(z)}}+a^{(z)},
\end{align}
with
\begin{itemize}
	\item $p_l^{(z)}=\left(u_l^{(z)}-{f}_l^{(z-1)}\right)^2 y_l^{(z)}$, \\$q_l^{(z)}=\left({f}_l^{(z-1)}-o_l^{(z)}\right)^2 y_l^{(z)}$.
	\vspace{2mm}
	\item $y^{(z)}_l=\frac{\rho^{(z)}}{f_{\max}-f_{\min}}+1.001\left(\frac{\partial \tilde{g}}{\partial {f}_l}\left(\mathbf{f}^{(z-1)}\right)\right)^+\\+0.001\left(\frac{\partial \tilde{g}}{\partial {f}_l}\left(\mathbf{f}^{(z-1)}\right)\right)^-$.
		\vspace{2mm}
	\item $u_l^{(z)}=\begin{cases}
		{f}_l^{(z-1)}+f_{\max}-f_{\min}, & z\leq 2,\\
		{f}_l^{(z-1)}+\gamma^{(z)}\left(u_l^{(z-1)}-{f}_l^{(z-2)}\right), & z\geq 3,\\
	\end{cases}$ denotes the upper asymptote in the $z$-th iteration.
	\vspace{2mm}
\item $o_l^{(z)}=\begin{cases}
{f}_l^{(z-1)}-f_{\max}+f_{\min}, & z\leq 2,\\
{f}_l^{(z-1)}-\gamma^{(z)}\left({f}_l^{(z-2)}-o_l^{(z-1)}\right), & z\geq 3,\\
\end{cases}$ denotes the lower asymptote in the $z$-th iteration.
	\vspace{2mm}
\item $a^{(z)}=\tilde{g}\left(\mathbf{f}^{(z-1)}\right)-\sum_{l=1}^{L} \frac{p^{(z)}_l}{u_l^{(z)}-{f}_l^{(z-1)}}-\frac{q^{(z)}_l}{{f}_l^{(z-1)}-o_l^{(z)}}$.
	\vspace{2mm}
\item $\rho^{(z)}=\frac{1}{10L}\operatorname{abs}\left(\frac{\partial \tilde{g}}{\partial \mathbf{f}}\left(\mathbf{f}^{(z-1)}\right)\right)^T\left(\mathbf{f}_{\max}-\mathbf{f}_{\min}\right)$ is the conservative factor, which can control the conservative property of the MMA convex approximation.
	\vspace{2mm}
\item $f_{l, \mathrm{low}}^{(z)}=\max\{f_{\min}, o_l^{(z)}+0.1\left({f}_l^{(z-1)}-o_l^{(z)}\right), {f}_l^{(z-1)}-0.5\left(f_{\max}-f_{\min}\right)\}$ represents the lower bound of ${f}_l$ in MMA convex approximation subproblem.
	\vspace{2mm}
\item $ f_{l, \mathrm{up}}^{(z)}=\min\{f_{\max}, u_l^{(z)}-0.1\left(u_l^{(z)}-{f}_l^{(z-1)}\right), {f}_l^{(z-1)}+0.5\left(f_{\max}-f_{\min}\right)\}$ represents the upper bound of ${f}_l$ in MMA convex approximation subproblem.
	\vspace{2mm}
\end{itemize}
In addition, the first-order gradients of the objective function w.r.t. the optimized variable, i.e., $\frac{\partial\tilde{g}}{\partial \mathbf{f}}$, can be derived as \footnote{\textcolor{blue}{As presented in Remark \ref{re_4}, during the alternating optimization process, the entire matrix $\tilde{\boldsymbol{\Theta}}$ can be treated as the optimization variable, without separately considering the modulation frequency variables contained within it. Thus, it can be treated as a known system parameter during the optimization of modulation frequencies.}}
\begin{align}
	&\mathcal{I}_l\left(\frac{\partial \tilde{g}}{\partial \mathbf{f}}\right)=\sum_{k=1}^{K}2\mathcal{R}\left(\frac{\partial\boldsymbol{\vartheta}_k^H}{\partial {f}_l}\mathbf{D}_k\boldsymbol{\vartheta}_k-\frac{\partial\boldsymbol{\vartheta}_k^H}{\partial {f}_l}\mathbf{d}_k\right),
\end{align}	
which represents the $l$-th element of $\frac{\partial\tilde{g}}{\partial \mathbf{f}}$. In addition, the $v$-th element of  $\frac{\partial\boldsymbol{\vartheta}_k}{\partial {f}_l}$ is given as
	\begin{itemize}
\item $\mathcal{I}_{v}\left(\frac{\partial\boldsymbol{\vartheta}_k}{\partial {f}_l}\right)=\begin{cases}
		A_0j2\pi g\left(\frac{d^{\mathrm{r}k}_v}{c}-t\right)\times\\e^{j2\pi g{f}_l\left(\frac{d^{\mathrm{r}k}_v}{c}-t\right)+\phi_0}, & v\in\mathcal{V},\\
		0, & otherwise,
	\end{cases}$
\end{itemize}
where we have $\mathcal{V}\triangleq\{(r-1)SMN+\tilde{s}N+(m-1)SN+n\}_{m\in\mathcal{M}, n\in\mathcal{N}}$ with $r=\lceil l/S\rceil, \tilde{s}=\operatorname{mod}(l-1, S)$. Note that $d^{\mathrm{r}k}_v$ denotes the distance between the $v$-th element of the FD-RIS and user $k$, which can be mapped as $d^{\mathrm{r}k}_{i_z, i_y}$ easily.

Considering that the MMA subproblem \eqref{eq_modulation_frequency_GCMMA} is a simple optimization problem without QoS constraints, we can analytically derive the close-form expression of its optimal solution through setting the $\frac{\partial \widehat{g}}{\partial {f}_l}(\mathbf{f}|\mathbf{f}^{(z-1)})=0$. Specifically, $\frac{\partial \widehat{g}}{\partial {f}_l}(\mathbf{f}|\mathbf{f}^{(z-1)})$ can be derived as
\begin{align}
	\frac{\partial \widehat{g}}{\partial {f}_l}(\mathbf{f}|\mathbf{f}^{(z-1)})=\frac{p_l^{(z)}}{\left(u_l^{(z)}-{f}_l\right)^2}-\frac{q_l^{(z)}}{\left({f}_l-o_l^{(z)}\right)^2}.
\end{align}
Let $\frac{\partial \widehat{g}}{\partial {f}_l}(\mathbf{f}|\mathbf{f}^{(z-1)})=0$, we can obtain the optimal solutions of the problem \eqref{eq_modulation_frequency_GCMMA}, which is given by
\begin{align}\label{eq_MMA_opti_solution}
	f_l^\mathrm{opt}=\begin{cases}
		f_{l, \mathrm{up}}^{(z)}, & \breve{f}_l>f_{l, \mathrm{up}}^{(z)},\\
		\breve{f}_l, & f_{l, \mathrm{low}}^{(z)}\leq\breve{f}_l\leq f_{l, \mathrm{up}}^{(z)},\\
		f_{l, \mathrm{low}}^{(z)}, &	\breve{f}_l<f_{l, \mathrm{low}}^{(z)},
	\end{cases}
\end{align}
where
 $\breve{f}_l=\frac{u_l^{(z)}\sqrt{\frac{q_l^{(z)}}{p_l^{(z)}}}+o_l^{(z)}}{1+\sqrt{\frac{q_l^{(z)}}{p_l^{(z)}}}}$.

\begin{remark}
\normalfont{It is important to highlight that in order to ensure the objective function value is decreasing monotonically, the constructed approximation subproblem of the original problem in the $z$-th iteration should meet the following conditions \cite{razaviyayn2013unified}: $\tilde{g}(\mathbf{f}^{(z-1)})=\widehat{g}(\mathbf{f}^{(z-1)}|\mathbf{f}^{(z-1)})$, $\frac{\partial\tilde{g}(\mathbf{f})}{\partial\mathbf{f}}|_{\mathbf{f}=\mathbf{f}^{(z-1)}}=\frac{\partial\widehat{g}(\mathbf{f}|\mathbf{f}^{(z-1)})}{\partial\mathbf{f}}|_{\mathbf{f}=\mathbf{f}^{(z-1)}}$ and $\widehat{g}(\mathbf{f}^{(z)}|\mathbf{f}^{(z-1)})\geq \tilde{g}(\mathbf{f}^{(z)})$. Unlike the successive convex approximation (SCA) method, the constructed MMA approximated objective function can not always satisfy the third condition, as it does not serve as the upper bound for the objective function of the original optimization problem. In other words, the MMA convex approximation method does not assure a strictly decreasing trend of the objective function value during the iteration process.}
\end{remark}
\begin{center}
	\begin{tabular}{p{8.5cm}}
		\toprule[2pt]
		\textbf{Algorithm 3:}  GCMMA Algorithm for Solving Problem \eqref{eq_modulation_frequency}   \\
		\midrule[1pt]
		1: Initialize $\mathbf{f}^{(0, 0)}$, tolerance accuracy $\widehat{\epsilon}$; Set outer iteration\\\quad index $z=0$.\\
		2: \textbf{Repeat} \\
		3: \quad Set inner  loop iteration index $x=0$; Construct the\\\qquad MMA convex approximation subproblem of the origi-\\\qquad nal problem based on $\mathbf{f}^{(z, 0)}$ utilizing \eqref{eq_MMA_approxi}; Update\\\qquad $\mathbf{f}^{(z+1, 0)}$ utilizing \eqref{eq_MMA_opti_solution} and calculate $\widehat{g}(\mathbf{f}^{(z+1, 0)})$ and\\\qquad $\tilde{g}(\mathbf{f}^{(z+1, 0)})$. \\
		4: \quad \textbf{While $\widehat{g}(\mathbf{f}^{(z+1, x)}) <\tilde{g}(\mathbf{f}^{(z+1, x)})$ do}\\
		5: \qquad Update the conservative factor $\rho^{(z+1, x+1)}$ leveraging\\\qquad\quad \eqref{eq_conservative_factor_update} and construct the more conservative MMA sub-\\\qquad\quad problem with the updated $\rho$;\\
		6: \qquad Update $\mathbf{f}^{(z+1, x+1)}$ utilizing \eqref{eq_MMA_opti_solution} and  calculate the\\\qquad\quad objective value $\widehat{g}(\mathbf{f}^{(z+1, x+1)})$ and $\tilde{g}(\mathbf{f}^{(z+1, x+1)})$; Let\\\qquad\quad $x\leftarrow x+1$.\\
		7: \quad\textbf{end while}\\
		8: \quad Update $\mathbf{f}^{(z+1, 0)}$ with the obtained $\mathbf{f}^{(z+1, x+1)}$;  Let $z\leftarrow$\\\qquad$ z+1$.\\
		6: \textbf{Until} $|\tilde{g}(\mathbf{f}^{(z+1, 0)})-\tilde{g}(\mathbf{f}^{(z, 0)})|\leq \widehat{\epsilon}$.\\
		7: \textbf{Output:} the optimal $\mathbf{f}$. \\
		\bottomrule[2pt]
	\end{tabular}
\end{center}

In order to address this issue, an additional inner loop iteration is implemented to fine-tune the conservative factor $\rho$ for a more conservative convex approximation. This adjustment aims to satisfy the condition $\widehat{g}(\mathbf{f}^{(z)}|\mathbf{f}^{(z-1)})\geq \tilde{g}(\mathbf{f}^{(z)})$. Specifically, \textcolor{blue}{we can readily demonstrate that $\widehat{g}(\mathbf{f}|\mathbf{f}^{(z-1)})$ in \eqref{eq_MMA_approxi} is a monotonically increasing function w.r.t. the conservative factor $\rho$. Therefore, if the solution of the MMA subproblem in the $z$-th outer-loop iteration does not satisfy the third condition, an inner-loop iteration is employed to progressively increase the conservative factor until the third condition is fulfilled.} \textcolor{blue}{Fig. \ref{fig:conservative} illustrates how the conservative factor adjusts the conservativeness of the MMA convex approximation.} In the $(x+1)$-th inner loop iteration, the update rule of the conservative factor can be given by
\begin{align}\label{eq_conservative_factor_update}
	\rho^{(z, x+1)}\leftarrow
	\min\Big\{1.1\left(\rho^{(z, x)} +\delta^{(z, x)}\right), 10\rho^{(z, x)}\Big\},
\end{align}
where $\delta^{(z, x)}=\tilde{g}(\mathbf{f}^{(z, x)})-\widehat{g}(\mathbf{f}^{(z, x)})$.
  More details for the GCMMA algorithm are presented in \cite{svanberg2002class, svanberg2007mma}. Furthermore, we provide Algorithm 3 to present the solving process for the optimization problem \eqref{eq_modulation_frequency} utilizing the GCMMA method.

\begin{figure}[ht]
	\centering
	\includegraphics[scale=0.4]{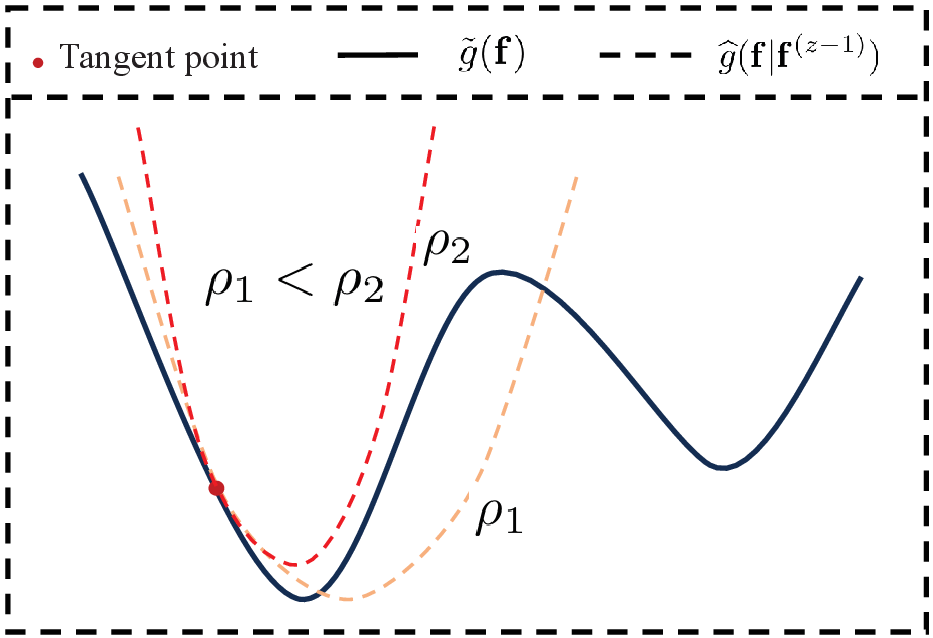}\\
	\caption{Schematic illustration of how the conservative factor $\rho$ adjusts the conservativeness of the MMA convex approximation.}\label{fig:conservative}
\end{figure}
\subsection{Analysis on computing complexity and convergence of the proposed algorithm}
The overall procedures of the proposed iterative algorithm for solving  the original optimization problem \eqref{eq_orig_opti} are summarized in Algorithm 4. This algorithm will iteratively solve the active beamforming, time-delay and modulation frequency subproblems until the difference in objective function values between adjacent iterations reach the minimum threshold $\breve{\epsilon}$.

Next, we will present the analysis on the computing complexity of the proposed algorithm. In particular, the computational complexity mainly comes from solving three key subproblems utilizing Algorithm 1-Algorithm 3, respectively. In terms of the active beamforming subproblem, the computing complexity is dominated by calculating the optimal $\mathbf{w}_k, ~k\in\mathcal{K},$ associated with the dual variable $\mu$, which is given by $\mathcal{O}_1=\mathcal{O}(I_1KN_\mathrm{t}^3)$, where $I_1$ denotes the total  number of iterations in Algorithm 1. For the second subproblem, the Riemannian conjugate gradient algorithm is leveraged to effectively tackle it. The computing complexity mainly stems from the calculation of the Euclidean gradient, which is derived as $\mathcal{O}_2=\mathcal{O}(I_2I^2)$ with $I_2$ being the total  number of iterations in Algorithm 2.  In the subproblem associated with the modulation frequencies, we utilize the GCMMA algorithm to solve it. The main computational challenge comes from the necessity to calculate the first-order gradient of the objective function w.r.t. the $\mathbf{f}$, leading to a computational complexity of $\mathcal{O}_3=\mathcal{O}(I_3I^2)$, where $I_3$ is the total number of outer loop iterations in Algorithm 3. Therefore, the total computational complexity of the proposed iterative algorithm can be calculated as $\mathcal{O}_\mathrm{tol}=I_\mathrm{tol}(\mathcal{O}_1+\mathcal{O}_2+\mathcal{O}_3)$, where $I_\mathrm{tol}$ is the total number of iterations for Algorithm 4.
\begin{center}
	\begin{tabular}{p{8.5cm}}
		\toprule[2pt]
		\textbf{Algorithm 4:}  The Proposed Iterative Algorithm for Solving the Original Problem \eqref{eq_orig_opti}   \\
		\midrule[1pt]
		1: Initialize $\Big(\{\mathbf{w}^{(0)}_k\}_{k=1}^K, \boldsymbol{\vartheta}^{(0)}, \mathbf{f}^{(0)}\Big)$, tolerance accuracy $\breve{\epsilon}$;\\\quad Set  the iteration index $q=0$.\\
		2: \textbf{Repeat} \\
		3: \quad Calculate the $W_k^{(q)}$ and $u_k^{(q)}$ with $\Big(\{\mathbf{w}^{(q)}_k\}_{k=1}^K, \boldsymbol{\vartheta}^{(q)},$\\\qquad$ \mathbf{f}^{(q)}\Big)$ leveraging \eqref{eq_opt_W_k_u_k}.\\
		4: \quad  Update $\{\mathbf{w}^{(q+1)}_k\}_{k=1}^K$ based on the solution solved by\\\qquad Algorithm 1.\\
		5: \quad Update the $W_k^{(q)}$ and $u_k^{(q)}$ with the $\{\mathbf{w}^{(q+1)}_k\}_{k=1}^K$.\\
		6: \quad Update $\boldsymbol{\vartheta}^{(q+1)}$ based on the solution solved by Algo-\\\qquad rithm 2.\\
		7: \quad Update the $W_k^{(q)}$ and $u_k^{(q)}$ with the $\boldsymbol{\vartheta}^{(q+1)}$.\\
		8: \quad Update $\mathbf{f}^{(q+1)}$ based on the solution solved by Algo-\\\qquad rithm 3; Let $q\leftarrow q+1$.\\
		9: \textbf{Until} the difference in objective function values between\\\quad adjacent iterations falls below $\breve{\epsilon}$.\\
		10: \textbf{Output:} the optimal $\Big(\{\mathbf{w}_k\}_{k=1}^K, \boldsymbol{\tau}, \mathbf{f}\Big)$. \\
		\bottomrule[2pt]
	\end{tabular}
\end{center}

The convergence of the proposed algorithm can be guaranteed due to the fact that the alternating algorithm can ensure that the current objective function value is not lower than that of the previous iteration during the iterative process since we can always obtain a solution not worse than the previous one. In addition, the convergence of the proposed algorithm will be further extensively examined through simulations in Section \ref{sec:S5}.
  \begin{figure*}[!t]
	\centering
	\includegraphics[scale=0.7]{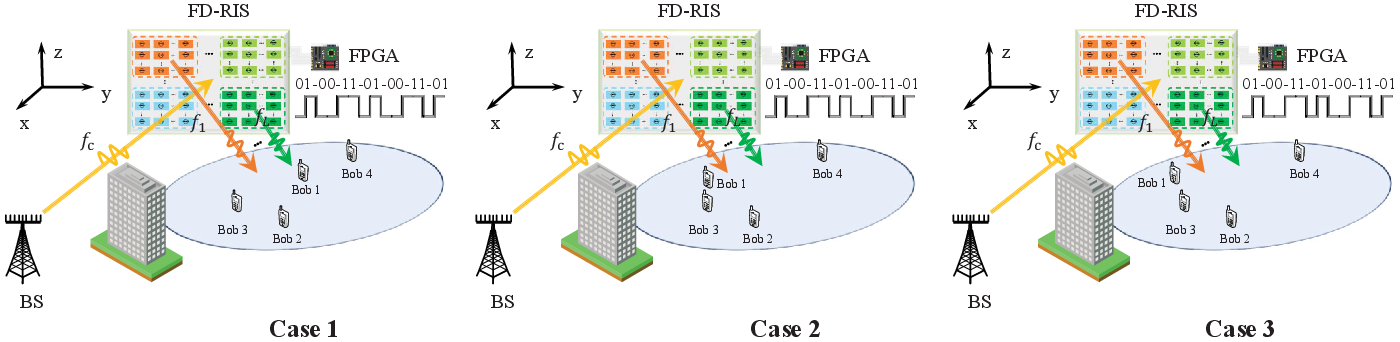}\\
	\caption{Three communication scenarios with various spatial correlations among users. }\label{fig:Three_cases}
\end{figure*}
\section{Numerical Simulations}\label{sec:S5}
In this section, the numerical simulations are implemented considering three scenarios featuring various spatial correlations among users to quantify the beamforming potential of the proposed multi-subarray FD-RIS and validate the effectiveness of the proposed iterative algorithm.  \textcolor{blue}{Note that the number of users is set as $K=$ 4. The specific details of these scenarios are provided as follows:}

\textcolor{blue}{\textbf{Case 1:} In this case, the spatial correlation among users is assumed to be negligible, which is a common assumption in conventional RIS-assisted communication systems. This setting corresponds to Case 1 illustrated in Fig. \ref{fig:Three_cases}. In particular, all users' positions are set as $(\theta_{\mathrm{r}1}, \theta_{\mathrm{r}2}, \theta_{\mathrm{r}3}, \theta_{\mathrm{r}4})=(100^\circ,110^\circ, 100^\circ, 120^\circ)$, $(\phi_{\mathrm{r}1}, \phi_{\mathrm{r}2}, \phi_{\mathrm{r}3}, \phi_{\mathrm{r}4})=(30^\circ,70^\circ, 30^\circ, 150^\circ)$, $(d_{\mathrm{r}1}, d_{\mathrm{r}2}, d_{\mathrm{r}3}, d_{\mathrm{r}4})$\\$=(40 m, 75 m, 55 m, 40m)$.
}

\textcolor{blue}{\textbf{Case 2:} In this scenario, we assume that there is a strong spatial correlation between Bob 1 and Bob 3 users. This setting corresponds to Case 2 illustrated in Fig. \ref{fig:Three_cases}. All users' positions are set as $(\theta_{\mathrm{r}1}, \theta_{\mathrm{r}2}, \theta_{\mathrm{r}3}, \theta_{\mathrm{r}4})=(100^\circ,110^\circ, 100^\circ, 120^\circ)$, $(\phi_{\mathrm{r}1}, \phi_{\mathrm{r}2}, \phi_{\mathrm{r}3}, \phi_{\mathrm{r}4})=(20^\circ,70^\circ, 30^\circ, 150^\circ)$, $(d_{\mathrm{r}1}, d_{\mathrm{r}2}, d_{\mathrm{r}3}, d_{\mathrm{r}4})$\\$=(40 m, 75 m, 55 m, 40m)$.
}

\textcolor{blue}{\textbf{Case 3:} It is assumed that Bob 1 and Bob 3 share the same spatial direction but are located at different distances, as illustrated in Case 3 of Fig. \ref{fig:Three_cases}. It is worth noting that, in this extreme scenario, the beamforming capability of a conventional RIS fails to effectively manage the mutual interference between Bob 1 and Bob 3. All users' positions are set as $(\theta_{\mathrm{r}1}, \theta_{\mathrm{r}2}, \theta_{\mathrm{r}3}, \theta_{\mathrm{r}4})=(100^\circ,110^\circ, 90^\circ, 120)$, $(\phi_{\mathrm{r}1}, \phi_{\mathrm{r}2}, \phi_{\mathrm{r}3}, \phi_{\mathrm{r}4})=(120^\circ,70^\circ, 30^\circ, 150^\circ)$, $(d_{\mathrm{r}1}, d_{\mathrm{r}2}, d_{\mathrm{r}3}, d_{\mathrm{r}4})$ $=(40 m, 75 m, 55 m, 40 m)$.
}

Additionally, we consider the following baseline schemes:
\begin{enumerate}
  \item \textbf{RIS-assisted scheme}: In this scheme, the conventional RIS comprising an equivalent number of elements as the FD-RIS will be employed to support the wireless communication networks. \textcolor{blue}{It is worth noting that in this scheme, both the active beamforming at the BS and the reflection coefficients of the RIS need to be jointly designed. For the active beamforming variables, the same method presented in Section IV-A is adopted to efficiently solve the problem. As for the passive beamforming variables, the Riemannian conjugate gradient algorithm is employed to achieve effective optimization.}
  \item \textbf{SDR scheme}: The SDR method \cite{luo2010semidefinite} and SCA technique are leveraged to address the active beamforming subproblem and time-delay subproblem. Specifically, letting $\mathbf{W}=[\mathbf{w}_1, \cdots, \mathbf{w}_K]$, $\widehat{\mathbf{W}}=\operatorname{vec}(\mathbf{W})\operatorname{vec}(\mathbf{W})^H$, the active beamforming design subproblem can be equivalently converted as
\begin{subequations}
	\begin{align}
		&\max _{\widehat{\mathbf{W}}}~ \sum_{k=1}^{K} \omega_k R_k,\notag \\
		&~\text { s.t. }\operatorname{Tr}(\widehat{\mathbf{W}})\leq P_\mathrm{tmax},\\
		&\qquad~\operatorname{rank}(\widehat{\mathbf{W}})=1,
	\end{align}
\end{subequations}	
where $R_k$ is re-written as
$R_k=\log_2\left(\frac{\operatorname{Tr}(\widehat{\mathbf{W}}\tilde{\mathbf{F}})+\sigma_k^2}{\operatorname{Tr}(\widehat{\mathbf{W}}\tilde{\mathbf{F}}_{-k})+\sigma_k^2}\right)$
with $\mathbf{F}=\mathbf{H}_\mathrm{BR}^H\boldsymbol{\Theta}_k^H\tilde{\boldsymbol{\Theta}}^H\mathbf{h}_{\mathrm{r}k}^*(\mathbf{H}_\mathrm{BR}^H\boldsymbol{\Theta}_k^H\tilde{\boldsymbol{\Theta}}^H\mathbf{h}_{\mathrm{r}k}^*)^H$, $\tilde{\mathbf{F}}=\mathbf{E}_{K}\otimes\mathbf{F}$, $\tilde{\mathbf{F}}_{-k}=(\mathbf{E}_{-k}\mathbf{E}_{-k}^T)\otimes\mathbf{F}$, \textcolor{blue}{$\mathbf{E}_{K}$ denotes a $K\times K$ identity matrix, $\mathbf{E}_{-k}$ is obtained by deleting the $k$-th column of $\mathbf{E}_{K}$.
Then, the first-order Taylor expansion is utilized to transform the non-convex objective function into convex. The rank-one constraint is equivalently reformulated and incorporated into the objective function as a penalty term. By iteratively enforcing the penalty term towards zero, the rank-one will be met.
 The detailed process refer to \cite{xiao2025STAR-RIS_UAV}.}
 \textcolor{blue}{ Similarly, we can also handle the time-delay subproblem as
 	\begin{subequations}
 		\begin{align}
 			&\max _{\mathbf{U}}~ \sum_{k=1}^{K} \omega_k R_k,\notag \\
 			&~\text { s.t. }\operatorname{diag}(\mathbf{U})= \mathbf{I}_{I\times 1},\\
 			&\qquad~\operatorname{rank}(\mathbf{U})=1,
 		\end{align}
 	\end{subequations}	
 	where $R_k=\log_2\left(\frac{\operatorname{Tr}(\mathbf{U}\tilde{\mathbf{G}}_k)+\sigma_k^2}{\operatorname{Tr}(\mathbf{U}\tilde{\mathbf{G}}_{-k})+\sigma_k^2}\right)$, $\mathbf{U}=\tilde{\boldsymbol{\theta}}^*\tilde{\boldsymbol{\theta}}^T$, $\boldsymbol{\theta}=\operatorname{diag}(\tilde{\boldsymbol{\Theta}})$, $\tilde{\mathbf{G}}_k=\mathbf{H}_{\mathrm{r}k}\boldsymbol{\Theta}_k\mathbf{H}_\mathrm{BR}\sum_{j=1}^{K}\mathbf{w}_j\mathbf{w}_j^H\mathbf{H}_\mathrm{BR}^H\boldsymbol{\Theta}_k^H\times$\\$\mathbf{H}_{\mathrm{r}k}^H,$ $\tilde{\mathbf{G}}_k=\mathbf{H}_{\mathrm{r}k}\boldsymbol{\Theta}_k\mathbf{H}_\mathrm{BR}\sum_{j\neq k}^{K}\mathbf{w}_j\mathbf{w}_j^H\mathbf{H}_\mathrm{BR}^H\boldsymbol{\Theta}_k^H\mathbf{H}_{\mathrm{r}k}^H$, $\mathbf{H}_{\mathrm{r}k}=\operatorname{Diag}(\mathbf{h}_{\mathrm{r}k})$. The same approach is employed to address this optimization problem as in the active beamforming subproblem above.
}
  \item \textbf{ZF scheme}: The ZF algorithm is adopted to derive the close-form active beamforming variables in this scheme. Specifically, $\mathbf{w}_k, ~k\in\mathcal{K}$, can be derived as
	$\mathbf{w}_k=\sqrt{P_k}\frac{\tilde{\mathbf{W}}(:, k)}{\|\tilde{\mathbf{W}}(:, k)\|}$,
where $\tilde{\mathbf{W}}=\mathbf{Z}^H(\mathbf{Z}\mathbf{Z}^H)^{-1}$ with $\mathbf{Z}=[(\mathbf{h}_{\mathrm{r}1}^T\tilde{\boldsymbol{\Theta}}\boldsymbol{\Theta}_1\mathbf{H}_\mathrm{BR})^T, \cdots, (\mathbf{h}_{\mathrm{r}K}^T\tilde{\boldsymbol{\Theta}}\boldsymbol{\Theta}_K\mathbf{H}_\mathrm{BR})^T]^T$. And $P_k$ denotes the allocated power for the $k$-th user, which needs to be carefully devised.
\end{enumerate}

The system parameter values of executing the simulations are set as: the carrier frequency $f_\mathrm{c}=28$ GHz, the FD-RIS position $(d_\mathrm{br}, \theta_\mathrm{br}, \varphi_\mathrm{br})$= $(100 m, 30^\circ, 120^\circ)$, $g=$1, the Rician factor $\beta_1=\beta_2=10$ dB, $\sigma_k^2=-110$ dBm, $k\in\mathcal{K}$, the number of elements in each subarray $M\times N=2\times 2$,  $f_{\max}=20$ MHz, $f_{\min}=0.2$ MHz.
\begin{figure}[ht]
	\centering
	\includegraphics[scale=0.33]{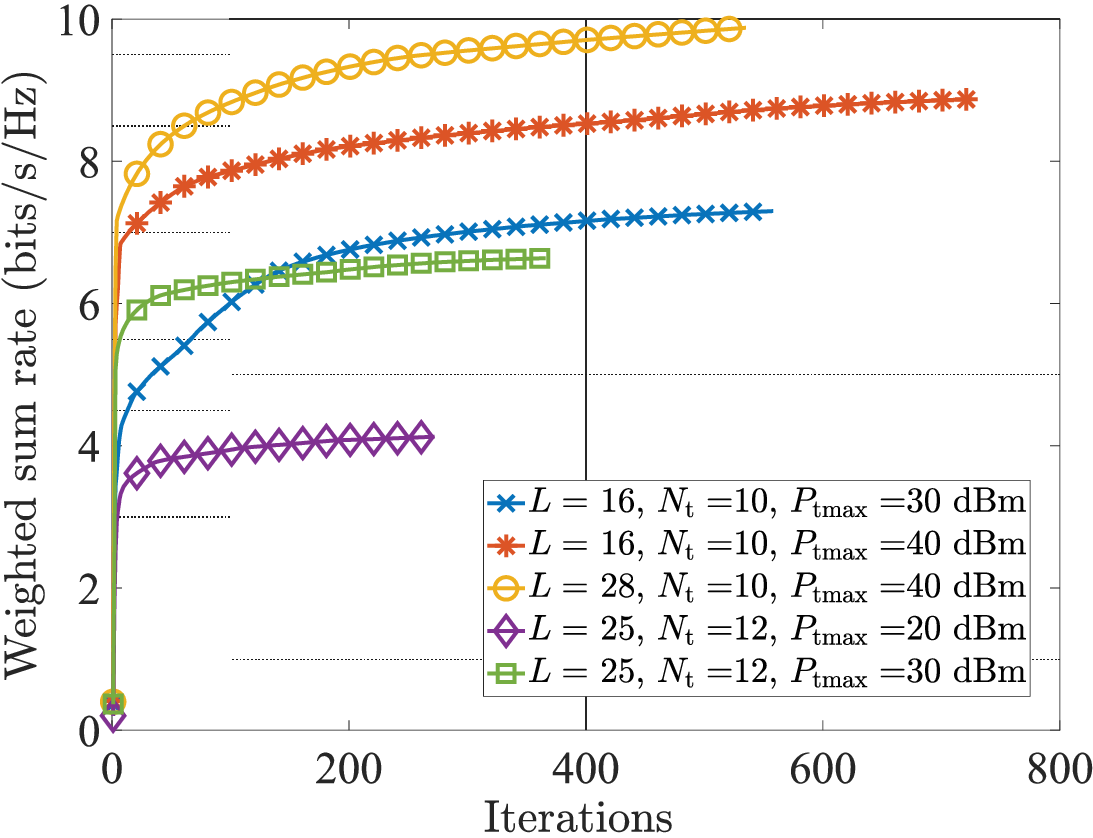}\\
	\caption{The convergence performance of the proposed algorithm taking account of different $L$, $N_\mathrm{t}$ and $P_\mathrm{tmax}$.}\label{fig:WS_vs_iterations}
\end{figure}
\begin{figure}[ht]
	\centering
	\includegraphics[scale=0.33]{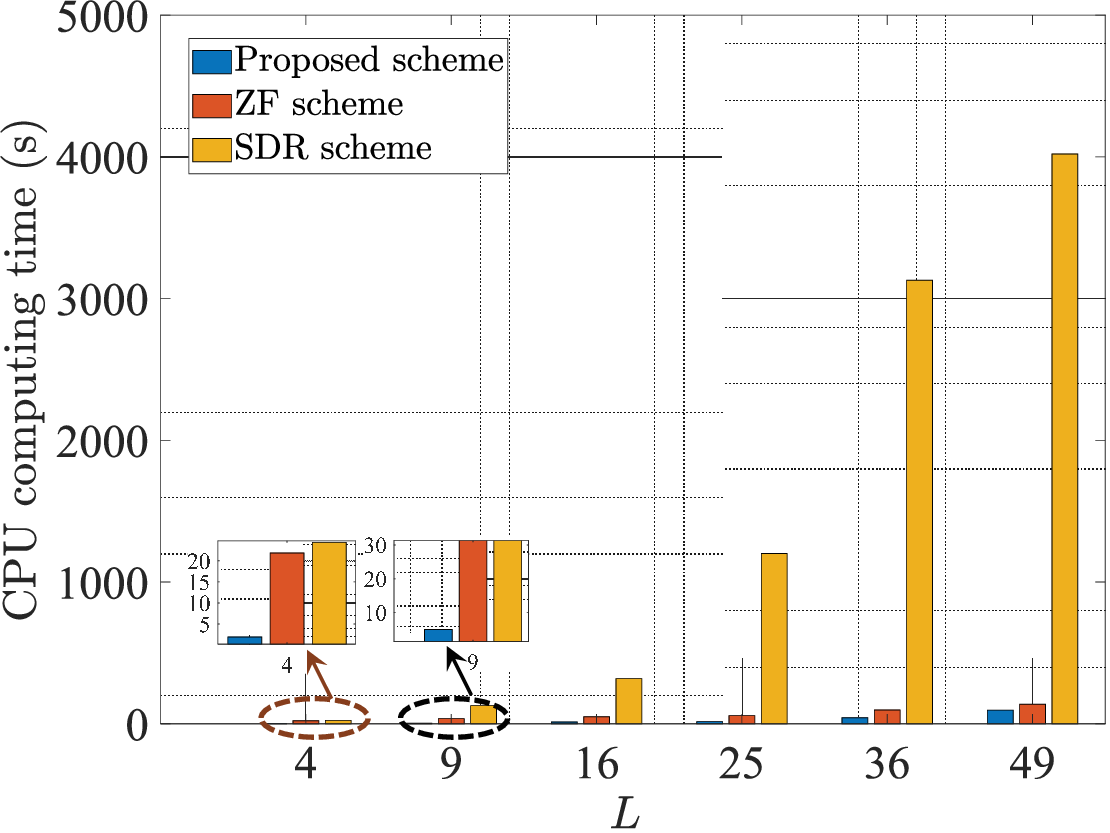}\\
	\caption{CPU computing time of the proposed algorithm and baseline algorithms versus the number of subarrays.}\label{fig:CPU_time}
\end{figure}

The convergence analysis of the proposed iterative algorithm is depicted in Fig. \ref{fig:WS_vs_iterations} considering different number of subarrays and antennas, as well as maximum transmitting power. Specifically, it is observed that the objective function value shows a monotonically increasing tendency with the growth of iterations across all cases. In addition, the algorithm under consideration consistently demonstrates the ability to swiftly attain a significantly high objective function value after just a few iterations, and typically converges within dozens of iterations across all scenarios, which indicates that the proposed algorithm's convergence can be guaranteed.

To further exhibit the advantage of  the proposed algorithm in computing efficiency, Fig. \ref{fig:CPU_time} illustrates a comparison of the CPU processing time between the proposed algorithm and the baseline algorithms, namely the ZF scheme and SDR scheme, taking account of  varying numbers of subarrays installed on the FD-RIS. The simulations are implemented with a CPU processor of Intel(R) Core(TM) i5-10505 CPU @ 3.20GHz. From this figure, we can find that  the time required by both the proposed scheme and the ZF scheme exhibits a gradual increase between $L=$ 4 and $L=$ 49. In contrast, the SDR scheme demonstrates a significant rise as the number of subarrays grows. Furthermore, the difference in CPU processing time  between the proposed method and the SDR-based scheme becomes increasingly pronounced, exceeding 3800 seconds when $ L =$ 49. The showcased performance of the proposed algorithm exhibits a notable edge in computational efficiency and the ability to tackle optimization hurdles in real-world communication systems.
\begin{figure}[ht]
	\centering
	\includegraphics[scale=0.45]{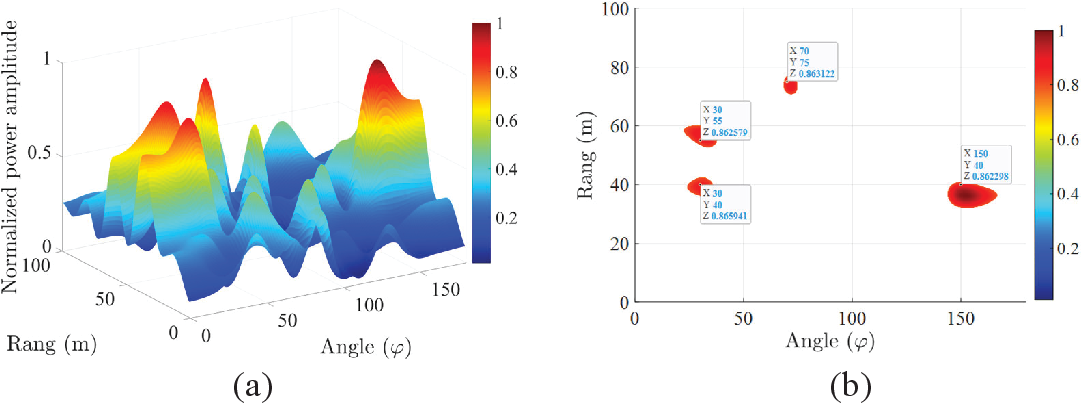}\\
	\caption{Users' receiving energy pattern with the assistance of the multi-subarray FD-RIS: (a) 3D pattern; (b) 2D pattern.}\label{fig:pattern}
\end{figure}

\textcolor{blue}{To clearly demonstrate and verify the powerful distance–angle modulation capability of the multi-subarray FD-RIS, the third scenario is adopted as an illustrative example, and the corresponding 3D and 2D patterns of the users’ received energy distribution assisted by the FD-RIS are shown in Fig. \ref{fig:pattern}. It is noted that, to facilitate pattern visualization, the elevation angles of all users are assumed identical and fixed at $90^\circ$.} Specifically, it is observed that the designed FD-RIS demonstrates the capability to manage the incoming signals in both distance and angle dimensions, which can precisely direct the signal energy towards the specific positions of users. Furthermore, to emphasize the considerable benefits of the FD-RIS in signal modulation across distance dimensions, we consider a scenario where two users are positioned at equal angles but varying distances, e.g., User 1 and User 3. The presented pattern results demonstrate that FD-RIS can accurately transmit the desired signal to the user's location even under scenarios where the users at same direction with different distances. This capability, which is unattainable for conventional RIS in far-field communication scenarios, highlights the unique advantages of the FD-RIS.
\begin{figure}[ht]
	\centering
	\includegraphics[scale=0.43]{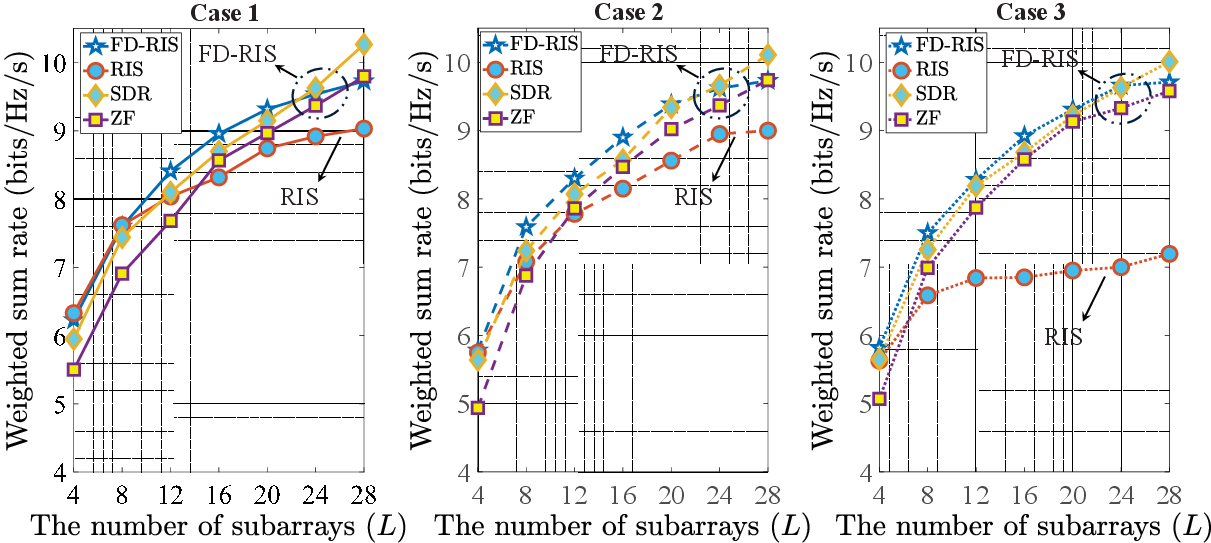}\\
	\caption{Weighted sum rate versus the number of subarrays $L$ considering  $N_\mathrm{t}=$ 10 and power budget $P_\mathrm{tmax}=$ 40 dBm at the BS, and $\omega_k=$ 0.25, $k\in\mathcal{K}$.}\label{fig:WS_vs_L}
\end{figure}

\textcolor{blue}{Fig.\ref{fig:WS_vs_L} gives the simulation results of the weighted sum rate versus the number of subarrays $L$ with $N_\mathrm{t}=$ 10 and power budget $P_\mathrm{tmax}=$ 40 dBm at the BS, and $\omega_k=$ 0.25, $k\in\mathcal{K}$, under the different communication cases.} In particular, we can find that the weighted sum rate shows a consistent gradual increase as the subarray size expands across all scenarios, but the increasing rates are decreasing. This is due to the fact that although the growing quantity of subarrays can enhance the level of control over incoming signals by providing additional degrees of freedom (DoFs), the system's performance may still be restricted by other system configurations. \textcolor{blue}{
 Additionally, the SDR solution exhibits slightly inferior performance compared to the proposed approach when 
 $L$ ranges from 4 to less than 24. However, as $L$ exceeds 24, the SDR scheme attains superior performance gains, albeit at the expense of substantially increased computational complexity required to solve the associated optimization problems. We can also find that the performance improvement of the proposed scheme significantly exceeds that of the ZF scheme. These results demonstrate that the proposed algorithm possesses a significant advantage in optimization ability.} 
  Furthermore, it is observed that the disparity in performance between the FD-RIS-assisted scheme and the conventional RIS-assisted scheme becomes more pronounced as the number of subarrays increases.  
 This is mainly because the FD-RIS provides higher DoFs, allowing signal energy to be more effectively focused on the users via distance–angle beamforming. \textcolor{blue}{According to the results presented in Case 1–Case 3, the performance gain provided by the RIS gradually decreases as the spatial correlation among users increases, especially when some users have the same spatial direction. This is because, as the spatial correlation among users gradually increases, the angle-only beamforming capability of conventional RIS is significantly weakened. In particular, when multiple users are located in the same or similar directions, the RIS can hardly distinguish signals among users through beam control in a single angular dimension, resulting in degraded beamforming performance. In such cases, the RIS can typically guarantee the QoS for only one user, while the QoS of other users becomes severely limited. In contrast, the FD-RIS-aided systems (including the FD-RIS, SDR, and ZF schemes) maintain nearly the same level of performance gain with the distance-angle beamforming, showcasing the high adaptability and robustness for transmission environment of FD-RIS. }

Then, we investigate the varied trend of the weighted sum rate versus the number of antennas $N_\mathrm{t}$ considering different $L$ and $P_\mathrm{tmax}$, and $\omega_k=$ 0.25, $k\in\mathcal{K}$, under three different communication scenarios, as depicted in Fig. \ref{fig:WS_vs_Nt}. Specifically, an increasing trend in the weighted sum rate is observed as the number of antennas increases in all scenarios, which is due to the fact that the increased $N_\mathrm{t}$ can provide more spatial DoFs to transmit different users' desired signal and effectively manage the interference among users' signals. \textcolor{blue}{
	Next, a more detailed analysis of Fig. \ref{fig:WS_vs_Nt} is presented. Specifically, according to the results shown in Case 1, it can be observed that in scenarios with low spatial correlation among users, when $L$ is small, both the FD-RIS  and the conventional RIS bring almost identical performance gains to the system. In this case, since the channels of different users are nearly independent, the beamforming capability of the conventional RIS is already sufficient to effectively enhance the desired signals and manage inter-user interference, meaning the additional dimensional advantage of the FD-RIS has not yet been fully utilized. However, as $L$ (for example, $L=$ 25), the FD-RIS-assisted system exhibits only a slight advantage over the conventional RIS in terms of weighted sum rate. This further confirms that under low spatial correlation, the system performance improvement of RIS is mainly constrained by the number of physical resources rather than the dimensionality of beamforming.	
	When the spatial correlation increases, the performance advantage of FD-RIS becomes increasingly significant. Specifically, as observed in Case 2 and Case 3, the FD-RIS maintains almost the same performance gain as in Case 1. In contrast, the conventional RIS, which only supports one-dimensional angular beamforming, fails to flexibly control signal energy and manage interference. As a result, its overall system performance deteriorates rapidly with increasing spatial correlation. This further validates the superior capability of FD-RIS's two-dimension beamforming.}

\begin{figure}[ht]
	\centering
	\includegraphics[scale=0.45]{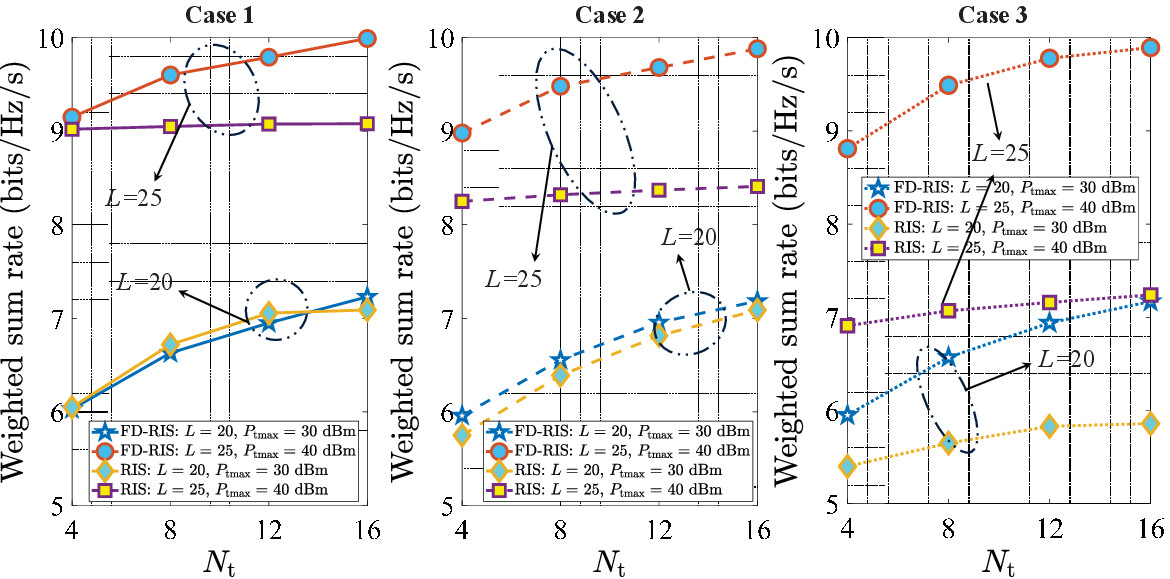}\\
	\caption{Weighted sum rate versus the number of antennas, $N_\mathrm{t}$, considering different number of subarrays and power budget at the BS, and $\omega_k=$ 0.25, $k\in\mathcal{K}$.}\label{fig:WS_vs_Nt}
\end{figure}
\begin{figure}[ht]
	\centering
	\includegraphics[scale=0.45]{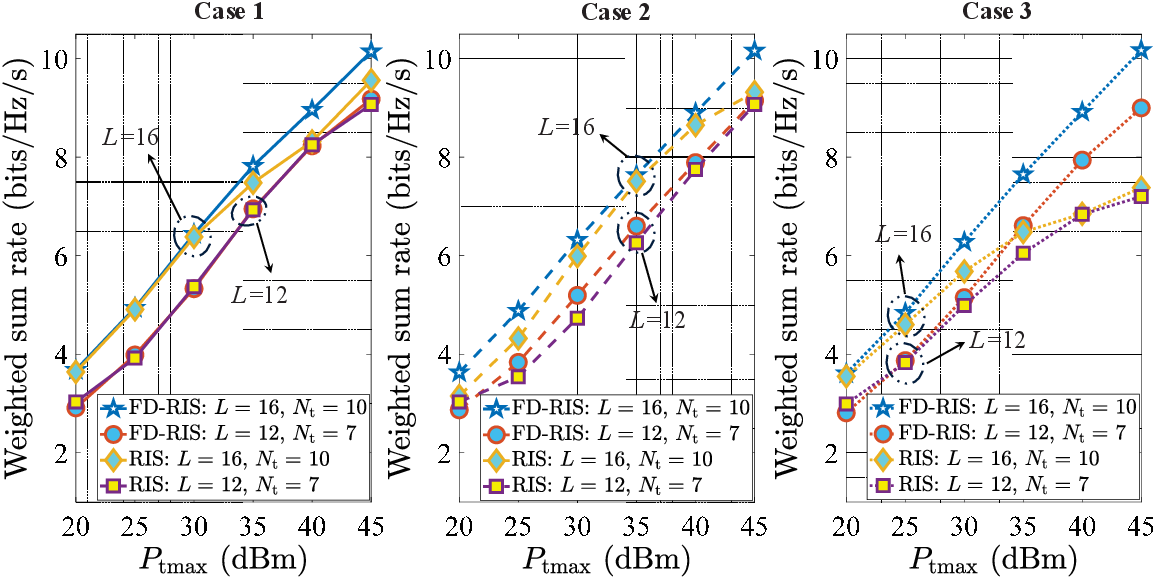}\\
	\caption{Weighted sum rate versus the power budget at the BS, $P_\mathrm{tmax}$, considering different number of antennas and subarrays, and $\omega_k=$ 0.25, $k\in\mathcal{K}$.}\label{fig:WS_vs_Ptmax}
\end{figure}

\textcolor{blue}{To quantify the energy efficiency of FD-RIS and RIS,} we further explore the influence of the power budget at the BS on the weighted sum rate taking into account of weight factor with $\omega_k=$ 0.25, $k\in\mathcal{K}$, and different $L$ and $N_\mathrm{t}$, as shown in Fig. \ref{fig:WS_vs_Ptmax}. \textcolor{blue}{In this simulation, two system configurations ($L=16, N_\mathrm{t}=10$ and $L=12, N_\mathrm{t}=7$) are employed to execute the proposed and baseline schemes. According to the presented simulation results, we can observe that the nearly linear growth of the weighted sum rate is observed with the expansion of the power budget $P_\mathrm{tmax}$ in the FD-RIS-aided schemes, while the increasing rate in the weighted sum rate gradually diminishes in the traditional RIS-supported scheme, particularly in scenarios with higher spatial correlation (i.e., Case 2 and Case 3). This is because the distance-angle beamforming capability of the FD-RIS enables it to effectively concentrate signal energy toward the intended users while mitigating mutual interference among them, thereby providing a significant advantage in enhancing the energy efficiency of wireless networks. Similarly, the FD-RIS maintains almost the same performance gain under different levels of spatial correlation, presenting remarkable environmental adaptability and robustness.}
%
\section{Conclusion and Prospect}\label{sec:S6}
Considering that the existing FD-RIS systems struggle with managing the undesired harmonic signals and exploiting the diversity of frequency offsets. To address these issues, this paper proposes a multi-subarray FD-RIS framework, in which the RIS is divided into multiple subarrays, each operating with a distinct time-modulation frequency to achieve frequency offset diversity. Additionally, to suppress the undesired harmonic interference, a novel time-modulation technique is introduced to periodically adjust the phase shifts of each element. According to the proposed multi-subarray FD-RIS framework, we first analytically derive its signal processing model. Then, the proposed multi-subarray FD-RIS is applied into the multi-user wireless communication systems to evaluate its potentials.
Specifically, we establish a non-convex optimization problem to maximize the weighted sum rate of all users through designing the active beamforming, time delays and modulation frequencies. To address this problem, \textcolor{blue}{the MMSE method is first leveraged to transformed the objective function and then a novel iterative algorithm based on Lagrange multiplier method, bisection search method, Riemannian conjugate gradient algorithm and GCMMA method is proposed to solve these three subproblem with low computing complexity. To evaluate the performance of the proposed scheme, three communication scenarios with different levels of spatial correlation among users are considered. Extensive simulation results indicate that the proposed multi-subarray FD-RIS scheme, supported by the distance–angle beamforming mechanism, can stably maintain overall capacity under various scenarios, demonstrating excellent robustness and adaptability. Moreover, compared with conventional RIS, the FD-RIS significantly enhances system performance through its powerful beamforming capability, with the advantage being particularly pronounced when users share the same spatial angle.}

\textcolor{blue}{This paper provides an initial contribution to the multi-subarray FD-RIS. For future works, several pertinent potential use cases can be recognized: (i) \textbf{ISAC}: Leveraging distance–angle beamforming of FD-RIS, the ISAC system can simultaneously obtain distance and angle information of targets based on the echo signals. 
	(ii) \textbf{Secure communication}: The single-dimensional control of traditional RIS may lead to ``security
	blind zone'' in communication systems. Specifically, when an eavesdropper and a legitimate user locate at the same spatial direction, the system struggles to distinguish between them, making it difficult to effectively suppress eavesdropping attempts. In contrast, FD-RIS can exploit differences in the distance domain to adjust the signal propagation, thereby overcoming this challenge.
}
\appendices
\section{Proof of Theorem \ref{th1}}\label{append 1}
In this section, the monotonicity of the function $P(\mu)=\sum_{k=1}^{K}\left\|\mathbf{w}^\mathrm{opt}_k(\mu)\right\|^2$ will be proved. Specifically, we first introduce two variables $\mu_1$ and $\mu_2$ with $\mu_1>\mu_2>0$. Furthermore, let $\{\mathbf{w}^\mathrm{opt}_k(\mu_1)\}_{k=1}^K$ and $\{\mathbf{w}^\mathrm{opt}_k(\mu_2)\}_{k=1}^K$ denote the optimal solutions of the Lagrange function \eqref{eq_lagrange} with $\mu_1$ and $\mu_2$, respectively. Considering that $\{\mathbf{w}^\mathrm{opt}_k(\mu_1)\}_{k=1}^K$ is the optimal solutions, we have
\begin{align}\label{eq_mu1}
	\mathscr{L}\left(\{\mathbf{w}^\mathrm{opt}_k(\mu_1)\}_{k=1}^K, \mu_1\right)\leq\mathscr{L}\left(\{\mathbf{w}^\mathrm{opt}_k(\mu_2)\}_{k=1}^K, \mu_1\right).
\end{align}
Similarly, we can also have
\begin{align}\label{eq_mu2}
	\mathscr{L}\left(\{\mathbf{w}^\mathrm{opt}_k(\mu_2)\}_{k=1}^K, \mu_2\right)\leq\mathscr{L}\left(\{\mathbf{w}^\mathrm{opt}_k(\mu_1)\}_{k=1}^K, \mu_2\right).
\end{align}
Adding \eqref{eq_mu1} and \eqref{eq_mu2}, we have
\begin{align}
	(\mu_1-\mu_2)P(\mu_1)\leq	(\mu_1-\mu_2)P(\mu_2).
\end{align}
Due to $\mu_1>\mu_2$, we can derive $P(\mu_1)\leq	P(\mu_2)$. Thus, we can prove that function $P(\mu)=\sum_{k=1}^{K}\left\|\mathbf{w}^\mathrm{opt}_k(\mu)\right\|^2$ is a monotonically decreasing function w.r.t. $\mu$.

\ifCLASSOPTIONcaptionsoff
  \newpage
\fi
\bibliographystyle{IEEEtran}
\bibliography{FD-RIS}

\end{document}